\documentclass[journal, romanappendices]{IEEEtran}

\usepackage{amssymb}
\usepackage{amsthm}

\usepackage[utf8]{inputenc}
\usepackage[T1]{fontenc}
\usepackage{url}
\usepackage{ifthen}
\usepackage{cite}
\usepackage{amsfonts,amsthm,bm}
\usepackage[cmex10]{amsmath}
\usepackage{comment}
\usepackage{enumerate}
\usepackage{mathtools}

\newenvironment{varsubequations}[1]
 {%
  \addtocounter{equation}{-1}%
  \begin{subequations}
  \renewcommand{\theparentequation}{#1}%
  \def\@currentlabel{#1}%
 }
 {%
  \end{subequations}\ignorespacesafterend
 }
	\newcommand{\equationVar}[2]{\begin{varsubequations}{#1-}\label{eq:#1}\begin{align}#2\end{align}\end{varsubequations}}

\usepackage{graphicx}
\newtheorem{theorem}{Theorem}
\newtheorem{proposition}{Proposition}

\newtheorem{lemma}{Lemma}
\newtheorem{corollary}{Corollary}
\newtheorem{remark}{Remark}

 \graphicspath{{./figures/}}
	\usepackage{cite}
	\usepackage[ruled,vlined,linesnumbered]{algorithm2e}

\newcommand{\Expect}{{\rm I\kern-.3em E}}

\newtheorem{example}{{\em Example}}

\newtheorem{prb}{Problem}

\author{Eleftherios Lampiris, Berksan Serbetci, Thrasyvoulos Spyropoulos, Giuseppe Caire, Petros Elia
\thanks{

E. Lampiris, B. Serbetci, T. Spyropoulos and P. Elia are with the Communication Systems Department of EURECOM, 06410 Sophia Antipolis, France, email: \{lampiris, serbetci, spyropou, elia\}@eurecom.fr.\nocite{maddah2014fundamental}

G. Caire is with the Communications and Information Theory Group (CommIT) of the Technical University of Berlin, 10587 Berlin, Germany, email: caire@tu-berlin.de.

E. Lampiris was previously with the Technical University of Berlin.

The work is supported by French National Research Agency (ANR) under the “5C-for-5G” JCJC project with reference number ANR-17-CE25-0001, by the ERC project DUALITY (grant agreement no. 725929), and by the ERC project CARENET (grant agreement no. 789190). Parts of this work have been published in the 18th International Symposium on Modeling and Optimization in Mobile, Ad Hoc and Wireless Networks (WiOpt 2020) \cite{serbetciAugmentingWiOpt2020}.}
}
\usepackage{xcolor}

\title{Multi-Transmitter Coded Caching Networks with Transmitter-side Knowledge of File Popularity}

\date{}

\begin{document}

\maketitle

\begin{abstract}
This work presents a new way of exploiting non-uniform file popularity in coded caching networks. Focusing on a fully-connected fully-interfering wireless setting with multiple cache-enabled transmitters and receivers, we show how non-uniform file popularity can be used very efficiently to accelerate the impact of transmitter-side data redundancy on receiver-side coded caching.
This approach is motivated by the recent discovery that, under any realistic file-size constraint, having content appear in multiple transmitters can in fact dramatically boost the speed-up factor attributed to coded caching.

We formulate an optimization problem that exploits file popularity to optimize the placement of files at the transmitters. We then provide  a proof that reduces significantly the variable search space, and propose a new search algorithm that solves the problem at hand. We also prove an analytical performance upper bound, which is in fact met by our algorithm in the regime of many receivers. Our work reflects the benefits of allocating higher cache redundancy to more popular files, but also reflects a law of diminishing returns where for example very popular files may in fact benefit from minimum redundancy.
In the end, this work reveals that in the context of coded caching, employing multiple transmitters can be a catalyst in fully exploiting file popularity, as it avoids various asymmetry complications that appear when file popularity is used to alter the receiver-side cache placement.
\end{abstract}

\section{Introduction}
In the context of cache-aided, interference-limited communication networks, the work of Maddah-Ali and Niesen~\cite{maddah2014fundamental} revealed how content that is properly placed at the caches of the receivers, can be used as side information to cancel interference and reduce delivery time. 

In particular, the work in~\cite{maddah2014fundamental} considers a single-antenna broadcast (downlink) configuration, where a transmitter has access to a library of $N$ files, each of size $F$ bits.
The transmitter serves---via a unit-capacity bottleneck link---a set of $K$ receiving users, each endowed with a cache of size $M\cdot F$ bits, corresponding to a fraction $\gamma\triangleq \frac{M}{N}$ of the library. The setting involves a \emph{cache-placement phase} where the caches are filled with content in a manner oblivious to future demands, and then a subsequent \emph{delivery phase} which starts with each user simultaneously demanding an independent file.

By exploiting \emph{content redundancy} where each bit of data can be placed at $KM/N = K\gamma$ users, the algorithm in \cite{maddah2014fundamental} could multicast different messages to $K\gamma + 1$ users at a time because each receiver could cancel the interference by accessing their own cache. This speedup factor of $K\gamma+1$ is commonly referred to as the Degrees-of-Freedom (DoF) performance, and it implies a worst-case delivery time\footnote{This is the normalized time that guarantees the successful delivery of all requested files, independent of the file-demand pattern.} equal to
\begin{equation}\label{eqMNdelay}
	T_{\text{MN}}=\frac{K(1-\gamma)}{1+K\gamma} \stackrel{K\to \infty}{=} \frac{1-\gamma}{\gamma}.
\end{equation}
The above delay was shown in~\cite{yuFactorOf2TransIT2019} to be within a multiplicative factor of at most $2.01$ from the information-theoretic optimal, and to be exactly optimal over the class of schemes that employ uncoded cache placement~\cite{wanUncodedPlacementTIT2020,yuExactUncodedTransIT2018}.

\paragraph*{Subpacketization and the redundancy constraint}
The above unbounded gain is in practice infeasible, mainly because it requires each file to be divided (subpacketized) into at least $K\choose K\gamma$ subfiles. Having files that do not scale exponentially in $K$, constitutes a prohibitive fundamental bottleneck \cite{shangguanHypergraphsTransIT2018,chittoorSubexponentialCCArXiv2020} which hard-bounds the DoF at very modest values\footnote{Some interesting progress on this, can be found in ~\cite{yanPDATransIT2017,tangSubpacketizationTransIT2018,shangguanHypergraphsTransIT2018,chittoorSubexponentialCCArXiv2020,shanmugamRuzsaGraphsISIT2017}.}.

A simple way to abide by the file-size constraint, is simply to assign the same cache content to entire groups of users (cf.\cite{shanmugamFiniteLengthTransIT2016}).
With this number of groups $\Lambda$ being constrained as
\begin{equation}\label{eqLambdaSubpacketization}
\binom{\Lambda}{\Lambda\gamma} \le F,
\end{equation}
the placement algorithm of \cite{maddah2014fundamental} is used to create $\Lambda$ different caches, and to assign the same cache to all the users belonging to the same group. Then, to satisfy the user demands, the delivery algorithm of \cite{maddah2014fundamental} --- which now enjoys a reduced DoF $\Lambda\gamma+1$ --- is repeated $\frac{K}{\Lambda}$ times, resulting in a delivery time of
\begin{equation} \label{eq:lambdaDelay}
	T_{\Lambda} = \frac{K(1-\gamma)}{1+\Lambda\gamma}.
\end{equation}

\paragraph*{Coded caching with transmitter-side cache redundancy}

As it turns out, the above subpacketization bottleneck is intimately connected, not only to the content redundancy $K\gamma$ at the receiver side, but also at the transmitter side.  This connection was made in~\cite{lampirisSubpacketizationJSAC} which --- in the context of multiple transmitting nodes (see~\cite{naderializadehFundamentalTransIT,shariatpanahiMultiServerTransIT,ZFE:15,piovanoOverloaded2017ISIT,zhangDelayedCSITtransIT2017,lampirisSubpacketizationCsitSPAWC,lampirisNoCsitISIT}) --- employed a novel fusion of coded caching and multi-antenna precoding, to dramatically reduce the subpacketization requirements of coded caching, and in the process to show that having multiple transmitter-side redundancy can in fact multiplicatively boost the caching gain.
In particular, for the coded caching scenario in~\cite{naderializadehFundamentalTransIT} (see also~\cite{shariatpanahiMultiServerTransIT}) where the $K$ receivers are served by $K_{T}$ transmitters each having access to a cache of normalized size $\gamma_{T}\in [ \frac{1}{K_{T}}, 1]$, the work in~\cite{lampirisSubpacketizationJSAC} showed that for $L\triangleq K_{T}\gamma_{T}$, and under the subpacketization of \eqref{eqLambdaSubpacketization} with $\Lambda\le \frac{K}{L}$, then one can get a dramatically decreased delivery time of
\begin{equation} \label{eq:DelayUniform}
	T = \frac{K(1-\gamma)}{L(1+\Lambda\gamma)}
\end{equation}
which is optimal under the assumption of uncoded placement~\cite{ParUnsEli2019}.

The above performance is achieved when each library file enjoys, on the transmitter side, an identical cache-redundancy equal to $L$, i.e. each file is cache at \textit{exactly} $L$ transmitters.
In our current work here, we propose and explore the endowing of some (generally more popular) files, with higher redundancy than their less popular counterparts. As we will see, this approach will not only improve performance, but will also allow us to utilize file popularity without breaking the symmetry of coded caching, as can often be the case when file popularity is used to alter the placement at the receiver side. This will become clearer below when we recall some existing methods of utilizing this knowledge.

\paragraph*{File popularity in coded caching, and the problem of symmetry}
Before recalling how non-uniform file popularity has been used in coded-caching, let us quickly recall that exploiting file popularity has been a key concept from the early works of Content Delivery Network systems~\cite{BorstCDN,sainoShardedToN2020}, Content-Centric and Information-Centric Networks~\cite{ZhangICN,carraElasticToN2020}, multi-tier networks \cite{alAbbasiMultiTierToN2019}, as well as in wireless edge caching works \cite{PaschosarticleJSAC} that followed the femto-caching ideas of \cite{shanmugamFemtocaching2013TransIT}. Such works generally focus on exploiting caches to `prefetch' content, and have little to do with using caches to handle network interference.  Even works that do consider PHY-aspects like multi-antenna beamforming, often assume that transmissions are, in essence, non-interfering~\cite{Tuholukova2017OptimalCA,AoPsounis2015}.

The connection between caching and interference management was mainly explored by works capitalizing on the interplay between coded caching and multiple transmitters, which initially focused on worst-case metrics, thus neglecting the effects of non-uniform file popularity.
Recently a variety of works such as \cite{niesenNonUniform17TransIT,zhangPopularity2018TransIT,hachemHeterogeneousPopularityInfocom2015,dingImprovedCachingICCC2017,OzfaturaCrossLevelICC2018,quintonNovelNonUniformCachingIZS2018,jinUncodedPlacementPopularityISIT2018,SaberaliOptimalNonUniformTransIT2019,jiOrderOptimalTIT2017,ramakrishnanEfficientITC2015,dengSubpacketizationNonUniformAsilomar2019,ChangHeterogeneousFilesDemandsISIT2019,AlLawatiNonIdenticalCWIT2017},
explored different ways of exploiting this popularity in the single-stream coded caching model.
As these efforts progressed, it was soon realized that incorporating file popularity in coded caching, brings to the fore a certain non-beneficial asymmetry which we discuss below.

In general, knowing the file popularity, would allow the grouping of similarly popular files, in order to allocate more cache space to popular files, thus leading to higher redundancy for more popular files and faster delivery.
Given, though, the multicast nature of coded caching, this approach brings to the fore the dilemma of whether or not multicast delivery messages should combine content from files that are dissimilar in terms of popularity.
This is an important dilemma with serious ramifications.
Choosing to not encode across different sub-libraries negates the very idea of coded caching, which benefits from encoding over as many users as possible. After all, as we have discussed, the gains of coded caching are proportional to how many users/files one encodes over. Instead, here, not encoding across sub-libraries, forces algorithms to separately deliver one sub-library after the other, which is a time-consuming process.
On the other hand, choosing to have popular and unpopular (sub) files coexist in a single transmission, can suffer from a certain asymmetry in the size of the composite subfiles. In principle, popular subfiles will tend to be longer than unpopular ones\footnote{This goes back to having designated more cache space for popular files, which often implies that the multicast messages will carry popular subfiles that are larger than their unpopular counterparts.}. This can in turn force very substantial zero-padding, which implies that only a fraction of the delivered bits actually corresponds to real content.

Drawing from the first paradigm, different works~\cite{niesenNonUniform17TransIT,zhangPopularity2018TransIT,hachemHeterogeneousPopularityInfocom2015} consider multicast messages (taken from~\cite{maddah2014fundamental}) which are composed of content from only one sub-library at a time. As was nicely shown in~\cite{zhangPopularity2018TransIT}, this approach provides for a bounded gap to the information-theoretic optimal\footnote{The multiplicative gap is approximately 50. Naturally the metric is the average delivery time, averaged over all demands.}. 
This gap in~\cite{zhangPopularity2018TransIT} was shown in~\cite{dingImprovedCachingICCC2017} to vanish for the special case of $K=2$.

Following the second paradigm, works such as~\cite{OzfaturaCrossLevelICC2018,quintonNovelNonUniformCachingIZS2018,jinUncodedPlacementPopularityISIT2018,SaberaliOptimalNonUniformTransIT2019} facilitate coding across sub-libraries, after optimizing the amount of cache each file can occupy as a function of its popularity. Interestingly, in some cases such as in~\cite{SaberaliOptimalNonUniformTransIT2019,zhangPopularity2018TransIT}, the optimization suggests --- under certain very important assumptions --- the need for only a very small number of sub-libraries. A similar conclusion was drawn in~\cite{dengSubpacketizationNonUniformAsilomar2019} for a decentralized setting\footnote{Coded caching placement strategies are divided into two broad categories, the centralized and the decentralized. Centralized placement, proposed in \cite{maddah2014fundamental}, assumes that the identity of the users is known during the placement phase and provides a deterministic caching strategy. In contrast, decentralized placement \cite{maddahDecentralizedToN2015} assumes that the identity of the users is unknown during the placement and as such the caching strategy is probabilistic, i.e., each chunk of a file is cached with a specified probability.}. Another interesting decentralized approach can be found in~\cite{jiOrderOptimalTIT2017} which combines a popularity-aware placement with a clique-cover delivery algorithm, to achieve --- under the assumption of a Zipf distribution (cf.~\cite{Newman2005Zipf}) and in the limit of large $K$ and large $K\gamma$ --- an order optimal performance\footnote{This means that, in the limit of infinite $K$, the gap to optimal is finite.}.
This performance was further improved in~\cite{ramakrishnanEfficientITC2015} which presented a delivery algorithm based on index coding (cf.~\cite{yossefIndexCodingTransIT2011}).

\subsection*{Current contribution: Boosting the impact of transmitter-side data redundancy using file popularity}

The dramatic impact of transmitter-side cache redundancy in coded caching, together with the aforementioned problem of symmetry, are two main motivating factors of our work. 
Focusing on a setting with $K_{T}$ cache-aided transmitters tasked with serving $K$ cache-aided receiving users, we explore the effect of allowing different files to experience different transmitter-side redundancies, depending on their popularity. In the context of coded caching, this constitutes a novel approach that allows us to benefit from a non-uniform file popularity, while having receivers that are agnostic to this popularity.

The main objective is to optimally divide the library into an arbitrary number of non-overlapping sub-libraries, and then to optimize the number of transmitters allocated to the files of any given sub-library.
As we will show in the main part of this work, finding the optimal solution to this placement proves to be a hard optimization problem. By solving this problem, we offer a multiplicative performance boost compared to the uniform popularity scenario (cf.~\eqref{eq:DelayUniform}), as well as a number of additional significant advantages compared to the state of art.

$\bullet$ A first advantage is that the receiver-side placement remains agnostic to file popularity. This allows the network to easily adapt to possible changes in file popularity, because updating the caches of a modest number of centralized transmitters is much easier than doing so for a large number of distributed receivers.

$\bullet$ Additionally, as we discussed earlier, popularity-aware receiver-side caching requires i) an accurate knowledge of the users that will be active during the placement phase and, ii) creates sub-file asymmetries which reduce the resulting gains.

$\bullet$ Finally, the adopted receiver-side placement strategy does not require the identity of the users to be known during the placement phase.

The work provides interesting insights. While it is beneficial to allocate higher cache redundancy to popular files (so that the majority of requests experience higher DoF performance), this has to be done with caution because after a certain point a law of diminishing returns kicks in. 
This is particularly true for very popular files, where---as we will see---a minimum redundancy is beneficial.

In the end, a key ingredient in our work is the fact that files do not have unbounded sizes. This may seem like an esoteric detail, but is in fact at the core of many coded caching problems. In our particular problem, having finite file sizes is what makes the impact of transmitter-side redundancy so powerful, and thus what motivates us to optimize this redundancy.

\paragraph*{Paper outline}In Section~\ref{secSystemModel} we present the system model and the notation. In Section~\ref{secCachingDeliveryPolicies} we discuss the caching and delivery algorithms as well as the optimization problem that we seek to solve. Further, in Section~\ref{secSolutionOverview} we first provide a proof that reduces significantly the variable search space and then we describe a novel algorithm that solves the optimization problem. In
Section~\ref{secSpeedingUp} we calculate a theoretical limit to the performance of our setting, while we prove that the reduced variable search space has the added benefit of providing an increased performance under any choice of variables.
Finally, in Section~\ref{sec:numev} we evaluate numerically the algorithm by plotting the multiplicative performance increase, compared to the uniform popularity case, as a function of the Zipf parameter $\alpha$ and for various number of users $K$.

\section{System Model \& Notation}\label{secSystemModel}

We consider the fully-connected, $K_{T}$-transmitter cache-aided setting, where $K_{T}$ single-antenna transmitters serve $K$ single-antenna receivers. Each transmitter and each receiver can store fraction $\gamma_{T}\in[\frac{1}{K_{T}},1]$ and fraction $\gamma\in[0,1]$ of the library, respectively. We assume that the library is comprised of $N$ files $W^1,W^2,\dots,W^N$, and that each file has size $F$ bits\footnote{This assumption is common in the literature, as non-equally sized files can be handled by making a content chunk the basic caching unit, as in~\cite{shanmugamFemtocaching2013TransIT}.} and is of finite size.
We assume that the system operates in the high Signal-to-Noise-Ratio region and that a single transmitter-to-receiver link has (normalized) capacity equal to one file per unit of time, as well as that the channel between any set of transmitters and receivers is of full rank with probability one\footnote{This requirement holds true in many wireless settings, as well as in wired settings with network-coding capabilities at the intermediate network nodes.}.

The caches of the transmitters and the receivers are filled with content during the placement phase.
During the delivery phase, each user will concurrently request a single file, and we assume that these requests follow a file popularity distribution that is known during the cache placement.
{In particular, we will focus on file popularity that follows the Zipf distribution~\cite{Newman2005Zipf} with parameter $\alpha>0$, under which the probability that file $W^n$ is requested, takes the form
\begin{align}
	p_{n} = \frac{n^{-\alpha}}{\sum_{k=1}^N k^{-\alpha}}, ~\forall n\in\{1,..., N\}.\label{eq:zipf}
\end{align}}

\subsubsection*{Notation} Symbols $\mathbb{N}, \mathbb{R}, \mathbb{C}$ denote the sets of natural, real and complex numbers, respectively. For $n,k\in\mathbb{N},~n\geq k$, we denote the binomial coefficient with $\binom{n}{k}$, while $[k]$ denotes the set $\{1,2, ..., k\}$.  We use $| \cdot |$ to denote the cardinality of a set. Bold letters are reserved for vectors, while for some vector $\mathbf{h}$, comprised of $Q$ elements, we denote its elements as~$h_{q}$,~$q\in[Q]$, i.e., $\mathbf{h}^{T} \triangleq [ h_{1}, h_{2}, ..., h_{Q}]$.

\section{Caching and delivery policies \& main problem\label{secCachingDeliveryPolicies}}

As suggested above, the caching policy at the transmitter-side is popularity-aware, while the receiver-side placement is not.
We begin with the general description of the transmitter-side caching policy and further describe the placement policy at the caches of the receivers.
Given these, we continue with the delivery algorithm, which is based on the algorithm of \cite{lampirisSubpacketizationJSAC}.
The last part of this section is dedicated to the presentation of the optimization problem that assigns content to the caches of the transmitters.

\subsection{Caching and delivery policies}

\subsubsection{Transmitter-side caching policy}

We segment the library into $Q$ non-overlapping sub-libraries.
Such segmentation is described via sets $\mathcal{B}_{q}\subset [N]$, $q\in[Q]$, and signifies that all files belonging to the same library would be assigned the same transmitter-side cache redundancy $L_{q}\in[1,K_{T}]$.
In other words, each file of sub-library $\mathcal{B}_{q}$ will be stored at exactly $L_{q}$ different transmitters.
As a consequence, variable $L_{q}$ is restricted to be in the range $[1,K_{T}]$.
On one end, we force each file to be cached by at least $1$ transmitter, hence allowing any request pattern to be satisfied in a finite time.
On the other end, the number of transmitters that can store a file is, naturally, limited by the number of different transmitters.
The above-described cache-redundancies need to satisfy the collective transmitter side cache-constraint
\begin{equation}\label{eqMemoryConstraint1}
	 \sum_{q=1}^{Q} |\mathcal{B}_{q} | \cdot L_{q} \le  N\cdot L,
\end{equation}
where for simplicity we use herein $L \triangleq K_{T}\gamma_{T}$.

In addition to the collective cache-constraint of \eqref{eqMemoryConstraint1}, the transmitter-side placement algorithm need also satisfy the per-transmitter cache-constraint of our model. In Appendix~\ref{sec:Redundancy} we propose an explicit algorithm which, for arbitrary $Q$, $\mathcal{B}$ and $\mathbf{L}$ satisfying~\eqref{eqMemoryConstraint1}, provides a placement which is based solely on the constraint in~\eqref{eqMemoryConstraint1}, while also satisfying the individual cache constraint.

\subsubsection{Receiver-side caching policy}
The receivers cache using a modified version of the algorithm of \cite{maddah2014fundamental}. Specifically, we create a set of $\Lambda < K$ different caches and assign one to each user in a round-robin manner. Variable $\Lambda$ is chosen such that $\Lambda\gamma$ is an integer and the subpacketization constraint is satisfied, i.e. $\binom{\Lambda}{\Lambda\gamma}\le F$. Each file $W^{n},~n\in[N]$, is split into $\Lambda\choose \Lambda\gamma$ equally-sized subfiles
\begin{equation}
	W^{n} \to \{ W^{n}_{\tau}, \ \ \tau\subset[\Lambda],|\tau|=\Lambda\gamma \}
\end{equation}
thus, each subfile has as index some set $\tau$, which is a $\Lambda\gamma$-sized subset of set $[\Lambda]$. Then, the $\ell^{\text{th}}$ cache takes the form
\begin{align}
	\mathcal{Z}_{\ell} = \big\{ W^{n}_{\tau} ~: ~ \ell\in\tau , \forall n\in[N]\big\},~\forall \ell\in\Lambda
\end{align}
which simply means that cache $\ell$ consists of all subfiles $W^{n}_{\tau}$, whose index $\tau$ contains $\ell$. The round-robin manner of assigning caches to users results in an approximate $\frac{K}{\Lambda}$ users to be assigned the same exact content.

\begin{example}
	Let us assume a setting comprized of $K=50$ users, each equipped with a cache of normalized size $\gamma=\frac{1}{10}$, and which users are divided into $\Lambda = 10$ groups.
	For example, such grouping would yield Group $1$ as $\mathcal{G}_{1} = \{1, 11, ..., 41\}$, Group $2$ as $\mathcal{G}_{2} = \{2, 12, ..., 42\}$ and so on.
	
	In the placement phase the files are divided into $\binom{\Lambda}{\Lambda\gamma}= 10$ subfiles as
$		W^{n} \to \{ W^{n}_{1}, W^{n}_{2}, ...,  W^{n}_{10}\},\ \  \forall n\in[N].$
Then, the contents of each cache would be
\begin{align*}
	\mathcal{Z}_{1}& = \{ W^{1}_{1}, ..., W^{n}_{1}\},
	\ \dots,\ 
	\mathcal{Z}_{10} = \{ W^{1}_{10}, ..., W^{n}_{10}\}.
\end{align*}
In the final step, each user of $\mathcal{G}_1$ is assigned cache $\mathcal{Z}_{1}$, each user of Group $2$ is assigned $\mathcal{Z}_{2}$ and so on.

\end{example}

\begin{table}[t]
\begin{tabular}{ll}
\hline
Parameters & Description \\ \hline
$N$    &      Number of different files      \\
$K$    &      Number of users       \\
$\gamma$ & Fraction of library each user can store\\
$\Lambda$ & Number of caches with different content\\
$K_T$      &      Number of single-antenna transmitters      \\
$\gamma_{T}$ & Fraction of library each transmitter can store\\
$n$ & File index\\
$p_n$      &    Probability that file $W^n$ will be requested\\
$\alpha$     & Zipf parameter\\
$Q$ & Number of sub-libraries\\
$\mathcal{B}_q$      &   Content of sub-library $q$\\
$\mathbf{n}$ & Vector storing the boundaries of the sub-libraries\\
$L_{q}$ & Number of transmitters caching file $W^n,$ $\forall n\in \mathcal{B}_{q}$\\
$\mathbf{L}$      &       Vector storing $L_{q}$\\
$K_{q}$ & Number of users requesting a file from $\mathcal{B}_{q}$\\
$\overline{K}_{q}$ & Expected number of users requesting a file from $\mathcal{B}_q$\\
$T(Q,\mathbf{n},\mathbf{L})$         &   Delay of expected requests as a function of $Q, \mathbf{n}, \mathbf{L}$\\
$T^{\star}$ & Min. expected delay optimized over all variables\\
$T^{\star}_{Q}$ & Min. expected delay optimized over $\mathbf{n},\mathbf{L}$. Fixed $Q$\\
$T^{\star}_{Q,\mathbf{n}}$ & Min. expected delay optimized over $\mathbf{L}$. Fixed $Q,\mathbf{n}$\\
$S_Q$ & Problem search space\\
$\pi_{q}$ & Sum probability of sub-library $\mathcal{B}_{q}$\\ \hline
  \noalign{\vskip 1mm}
\end{tabular}
\caption{Notation summary}
\label{tab:Not}
\end{table}

\subsubsection{Content delivery policy\label{sec:Delivery}}
The delivery phase begins with the concurrent request of any single file from each user.
The fulfilment of these requests happens in a per-sub-library manner. Specifically, for each set of $K_{q}$ users, requesting files from sub-library $\mathcal{B}_{q}$, we employ the algorithm of \cite{lampirisSubpacketizationJSAC}.

\subsection{Main optimization problem - Placement at the transmitters}\label{secMainOptProb}

Having described the caching policy at the users and the subsequent delivery policy it remains to design the caches of the transmitters such as to reduce the delivery time. To this end, we need to
\begin{itemize}
	\item select the number of sub-libraries $Q$,
	\item segment the library into $\mathcal{B}_{q} \subset[N], {q\in[Q]}$, and
	\item associate a cache redundancy $L_{q}$ with each $\mathcal{B}_{q}$.
\end{itemize}

Since the request pattern is of a stochastic nature we will focus on minimizing the delivery time of the expected demand.
In other words, we assume that the number of users requesting a file from sub-library $\mathcal{B}_{q}$ is $K_{q} = \overline{K}_{q} =K \pi_{q} $, where we denote the cumulative probability of the files of sub-library $\mathcal{B}_q$, $q\in[Q]$ by $\pi_{q}\triangleq\sum_{k\in\mathcal{B}_q} p_{k}$ .

Taking the above into account, the delivery time for each sub-library takes the form
\begin{equation}\label{eq:qDelay}
	T_{q} = \frac{K_q(1-\gamma)}{ \min\{ L_q(1+\Lambda\gamma), K_{q}\} }, \ \ q\in[Q]
\end{equation}
where the minimum in the denominator describes that the number of users served in a given time-slot is upper bounded by the number of available users.

In addition, in order to magnify the impact of transmitter-side cache redundancy, we impose a further constraint on the value of $L_q$.
{Specifically, we force $L_{q} \le \frac{K_{q}}{\Lambda}$, which ensures that the achieved DoF is always a multiple of $L_{q}$, i.e. takes the form $L_{q}( \Lambda\gamma+1)$ for any value of $L_q$ (cf.~\cite{lampirisSubpacketizationJSAC}). Beyond this value of $L_q$ the best known results achieve only an additive gain i.e., increasing the DoF by $1$ for each increase of $L_q$ by $1$, while negatively affecting the subpacketization~\cite{parrinelloExtendingISIT2020}.}
\begin{remark}
	While treating demands in a per sub-library manner is not necessarily optimal we note that, at the time of this writing, no known delivery algorithm can merge demands from multiple libraries in a single transmission.
	In particular, to date, no known multi-transmitter coded caching algorithm can improve the current performance we achieve, by simultaneously transmitting files that have different transmitter-side redundancy. We believe this to be an interesting open problem.
\end{remark}

Combing the delivery time of each sub-library we get the achievable delay of
\begin{equation}\label{eqDTinit}
	T  =  \sum_{q =1}^{Q}  \frac{K \pi_{q} (1-\gamma)}{\min\{ L_q(1+\Lambda\gamma), K \pi_{q}\} }.
\end{equation}

We can further improve the delivery time of \eqref{eqDTinit} by considering that a set of ultra popular files may be requested by a significant amount of users, hence these files can be naturally multicasted from a single antenna, i.e. to be communicated sequentially and without employing coded caching techniques. This would allow to serve a significant number of users with minimal transmitter-side resources, since storing each file at a single transmitter would suffice to satisfy such demands.
We place these files in sub-library $\mathcal{B}_{1}$, while noting that this additional (natural multicasting) option does not limit the optimization range because $\mathcal{B}_1$ could be---if indicated by the optimization---empty.
Consequently, the cache redundancy assigned to this sub-library is $L_{1}=1$, and the respective delay is $T _1= | \mathcal{B}_{1} |$ and corresponds to broadcasting the whole content of the sub-library.

Combining the above-described delivery delays of each sub-library, and for simplicity refraining from displaying the minimum function, the overall delay achieved takes the form
\begin{equation}\label{eqDeliveryTime1}
	 T(Q, \mathcal{B}, \mathbf{L}) = | \mathcal{B}_{1} | + \sum_{q =2}^{Q}  \frac{K\pi_{q} (1-\gamma)}{  L_{q}(1+\Lambda\gamma)}.
\end{equation}

Because \eqref{eqDeliveryTime1} is linearly dependent on the number of users requesting a file from each sub-library, we can conclude that the expected delay is equal to the delay of the expected demand, i.e. $\overline{K}_{q} = K \pi_{q}$.

Thus, the optimization problem at hand is expressed as
\begin{prb}[General Optimization Problem]\label{eqOriginalProblem}
\equationVar{P1}{
\underset{{Q, \mathcal{B},\mathbf{L}}} {\mathrm{min.}}~ \mathbb{E}\{ &T(Q, \mathcal{B},\mathbf{L})\} \label{originalOptimizationProb}\\
\mathbf{s.t.}\quad & Q \in [N] \label{eqQconstraintOriginal}\\
&|\mathcal{B}_{1}| + \sum_{q= 2}^Q L_q | \mathcal{B}_{q}| \leq L N, \label{EQnConstraintOriginal}\\
& L_{q} \in  \left[ 1,  U_q \right],~~\forall q \in [Q]. \label{eqLambdaConstraintOriginal}
}
\end{prb}
where
$
U_q = \min\left\{K_T, {{K}\pi_q}/{\Lambda}\right\}$.

\section{Description of the Optimization Algorithm}
\label{secSolutionOverview}

Retrieving the optimal solution of Problem~\ref{eqOriginalProblem} requires optimizing variables $Q,\mathcal{B}, \mathbf{L}$.
The main difficulty we face is that the complexity increases exponentially for non-trivial values of $Q$.
This high complexity is attributed to the need to segment set $[N]$ into $Q$ non-overlapping subsets, which have the property of minimizing the problem at hand.
For example, for $Q=2$ the search space for $\mathcal{B}$ has size $2^{N}$, due to the need to consider every possible subset size for sub-library $\mathcal{B}_{1}$, i.e.
\begin{align}
	\binom{N}{1} + 	\binom{N}{2} + 	... + 	\binom{N}{N} = 2^N.
\end{align}

In the general case, the size of the search space is exponential in $N$, as we show in the following proposition.
\begin{proposition}
	The size of the search space for determining sub-libraries $\mathcal{B}$ in Problem~\ref{eqOriginalProblem}, takes the form
	\begin{equation}\label{eqSearchSpaceProp}
		| S_1(Q) | = Q^N, \ \ Q\in [N].
	\end{equation}
\end{proposition}
\begin{proof}
	We begin with reminding the binomial equation which holds for any $x,y\in\mathbb{C}$,
	\begin{equation}\label{eqBinomial}
		(x+y)^{n} = \sum_{k=0}^{n} \binom{n}{k} x^{n-k}y^{k}.
	\end{equation}
For any $Q\in[N]$ we denote the size of the first sub-library with $k_1\in[N]$, and the size of subsequent sub-libraries with $k_q\in[\mathcal{N}_{q}]$, where $\mathcal{N}_{q} = N-\sum_{i=1}^{q-1} k_i$ signifies the maximum number of elements in sub-library $\mathcal{B}_{q}$.
Hence, for some $\mathcal{N}_{Q-1}$, the number of possible sub-libraries $\mathcal{B}_{Q-1}$ are
\begin{equation}
	\sum_{k_{Q-1}=1}^{\mathcal{N}_{Q-1}}\binom{\mathcal{N}_{Q-1}}{k_{Q-1}}.
\end{equation}

Using the above, we can continue to calculate all possible pairs $\mathcal{B}_{Q-2}, \mathcal{B}_{Q-1}$, under the assumption that we have allocated some $\mathcal{N}_{Q-2}$ files to the first $Q-3$ sub-libraries.
The number of all possible sub-library pairs $\mathcal{B}_{Q-2}, \mathcal{B}_{Q-1}$ is
\begin{equation}
	\sum_{k_{Q-2}=1}^{\mathcal{N}_{Q-2}} \left( \binom{\mathcal{N}_{Q-2}}{k_{Q-2}} \cdot \sum_{k_{Q-1}=1}^{\mathcal{N}_{Q-1}}\binom{\mathcal{N}_{Q-1}}{k_{Q-1}}\right).
\end{equation}

Extending this to the general case yields
\begin{align}\nonumber
	|S_Q|& = \sum_{k_1=1}^{N}\Bigg\{ \binom{N}{k_{1}}\sum_{k_2=1}^{\mathcal{N}_{1}} \bigg[ \binom{\mathcal{N}_{1}}{k_{2}} \cdots\\	&\cdots\sum_{k_{Q-2}=1}^{\mathcal{N}_{Q-2}} \left( \binom{\mathcal{N}_{Q-2}}{k_{Q-2}}
	\sum_{k_{Q-1}=1}^{\mathcal{N}_{Q-1}} \binom{\mathcal{N}_{Q-1}}{k_{Q-1}}\right) \cdots \bigg] \Bigg\}.\label{eqSumBinomials}
\end{align}
	Using \eqref{eqBinomial} for $x=y =1$, the last summand of~\eqref{eqSumBinomials} becomes $2^{\mathcal{N}_{Q-2}}$. Similarly, the inner most parenthesis (two last summands) of \eqref{eqSumBinomials} can be calculated using the newly acquired value of the last summand and \eqref{eqSumBinomials} as follows
	\begin{equation*}
		\sum_{k_{Q-2}=1}^{\mathcal{N}_{Q-2}} \binom{\mathcal{N}_{Q-2}}{k_{Q-2}}2^{\mathcal{N}_{Q-2}} = 3^{\mathcal{N}_{Q-3}}.
	\end{equation*}
Continuing in the same manner yields \eqref{eqSearchSpaceProp}.
\end{proof}

We employ the following steps to solve Problem~\ref{eqOriginalProblem}.
\begin{enumerate}
	\item In Lemma~\ref{lemmaConsecutive} (Section~\ref{secReducedSpace}) we prove that the optimal solution of Problem~\ref{eqOriginalProblem} should be of the form  $\mathcal{B}_{q} = \{ n_{q-1}\!+\!1, ..., n_{q}\},$ $\forall q\in[Q]$, where $n_{0}=0$ and $n_{Q}=N$. In other words, the first sub-library should be comprised of the $n_{1}$ most popular files, the second sub-library would contain files $\{n_{1}+1, ..., n_{2}\}$, and so on.
	From this point on we refer to a library segmentation using vector
	\begin{equation}\label{eqNvec}
		\mathbf{n} \triangleq \{n_{1}, ..., n_{Q}\!=\! N\}.
	\end{equation}
	
	We can easily deduce the size of the reduced search space
	\begin{equation}\label{eqFinalSearchSpace}
		|S_{2}(Q) | = \binom{N}{Q} \approx \left( \frac{N}{Q}\right)^{Q}
	\end{equation}
which considerably prunes the search space from exponential in $N$ to polynomial in $N$, without sacrificing optimality.
\item We provide an algorithm that searches $S_{2}(Q)$ requiring complexity at most
\begin{equation}
	\left( \log_{2}N\right)^{Q}.
\end{equation}
	\item We reformulate the objective function as a set of nested problems, as follows
	\begin{equation}
		 \underset{Q(\mathbf{n}^{\star}, \mathbf{L}^{\star})}{\mathrm{min.}}~~\underset{\mathbf{n}(\mathbf{L}^{\star})}{\mathrm{min.}}~~ \underset{\mathbf{L}}{\mathrm{min.}} ~~\mathbb{E}\{ T(Q, \mathbf{n},\mathbf{L})\}
	\end{equation}
	which effectively means that for each search of the outmost variables $Q$ and $ \mathbf{n}$ we optimize the innermost variables $\mathbf{n}, \mathbf{L}$ and $\mathbf{L}$, respectively
	hence, maintaining the optimality of the solution \cite{BoydConvex2004}.
\item As we show in Section~\ref{secLagrangeMultipliers}, calculating the optimal $\mathbf{L}$ can be achieved via the use of the Karush-Kuhn-Tucker (KKT) conditions.
In other words, the innermost problem has an analytical solution, conditional on the values of $Q$ and $\mathbf{n}$, which can be used directly for the outer optimization of these variables.
Furthermore, the continuous relaxation of $\mathbf{L}$ required by the application of the KKT conditions would result in a small performance degradation, which we show in Lemma~\ref{lemmaRelaxationL} is at most $12\%$.

\item Finally, we prove that the objective function, when optimized over both $\mathbf{L}$ and $\mathbf{n}$ is monotonically decreasing when $Q\in[1,Q^{\star}]$ and monotonincally increasing for $Q\in[Q^{\star},N]$. Thus, the function has a single minimum point which we calculate using a bisection algorithm.
\end{enumerate}

\subsection{Reduced sub-library search space\label{secReducedSpace}}

In order to show the optimality of the solution when considering the reduced sub-space in \eqref{eqNvec}-\eqref{eqFinalSearchSpace}, we begin with a corollary that describes the relationship between the cache-allocation among any two arbitrary sub-libraries.
\begin{corollary}\label{corMoreNeedsMore}
	For two arbitrary sub-libraries $\mathcal{B}_{q}, \mathcal{B}_{r} \subset[N]$ for which $\pi_{q}>\pi_{r}$, their respective optimal cache-redundancy allocations satisfy
	\begin{equation}
		L_{q}^{\star} > L_{r}^{\star}.
	\end{equation}
\end{corollary}
\begin{proof}
	The proof is relegated to Appendix~\ref{proofMoreNeedsMore}.
\end{proof}

With this in place, we proceed with the lemma that establishes the optimality of the consecutively indexed library segmentation.
\begin{lemma}\label{lemmaConsecutive}
	 For arbitrary number of sub-libraries $Q$, the sub-libraries producing the optimal delay are those whose files have consecutive indices.
\end{lemma}
\begin{proof}
The proof is relegated to Appendix~\ref{proofConsecutive}.
\end{proof}

Consequently, using Lemma~\ref{lemmaConsecutive} we can simplify the objective function as
\begin{equation}\label{eqDeliveryTime2}
	 \mathbb{E} \big\{ T(Q, \mathbf{n},\mathbf{L}) \big\}  = {n}_{1}+ \sum_{q =2}^{Q}  \frac{\overline{K}_{q} (1-\gamma)}{L_{q}(1+\Lambda\gamma)}
\end{equation}
and the optimization problem takes the following form.
\begin{prb}[Main Optimization Problem]\label{prb:mainprb}
\equationVar{P2}{
&\underset{Q,\mathbf{n},\mathbf{L}}{\mathrm{min.}}~ \mathbb{E} \big\{ T(Q, \mathbf{n},\mathbf{L}) \big\}\label{mainoptfunc}\\
\mathbf{s.t.}\quad & Q \in [N] \label{eqQconstraint}\\
&n_{1} + \sum_{q= 2}^Q L_q (n_q - n_{q-1}) \leq L N, \label{EQnConstraint}\\
& L_{q} \in  \left[ 1,  U_q \right],~~\forall q \in [Q]. \label{eqLambdaConstraint}
}
\end{prb}

The constraints of Problem~\ref{prb:mainprb} are those of Problem~\ref{eqOriginalProblem}, with the notable difference being constraint \eqref{EQnConstraint} which substitutes \eqref{EQnConstraintOriginal}, to yield a substantially reduced search space without loss of optimality.

As we discuss in Section~\ref{secSpeedingUp}, the library segmentation of Problem~\ref{prb:mainprb} has an added benefit, on top of reducing the search space, compared to the general library segmentation of Problem~\ref{eqOriginalProblem}. We show that \textit{any} library segmentation as the one proposed in Problem~\ref{prb:mainprb}, and under the optimal allocation of cache-redundancies $L_{q}$, would outperform the uniform popularity setting.
On the other hand, the general library segmentation of Problem~\ref{eqOriginalProblem} does not share this property (see discussion in Corollary~\ref{corImprovement} and Remark~\ref{remarkNonOptimalLibraries}).

\subsection{Optimizing cache redundancies $L_q$}\label{secLagrangeMultipliers}

We begin this section with the following lemma.
\begin{lemma}
The objective function is convex in variables $\mathbf{L}$ for fixed $Q$ and $\mathbf{n}$.
\end{lemma}
\begin{proof}
The proof is relegated to Appendix~\ref{sec:BiconvexProof}.
\end{proof}

Hence, applying the KKT condition would provide the optimal vector $\mathbf{L}$.
The Lagrangian takes the form
\begin{align}
\mathcal{L} = &n_1 + \sum_{q=2}^Q \frac{{K}\pi_{q}(1-\gamma)}{L_q(1+\Lambda\gamma)}\nonumber\\
&+ \lambda \left(n_{1}\! +\! \sum_{q=2}^{Q} L_{q}(n_{q}-n_{q-1}) - L N\right)\nonumber\\
&+ \sum_{q=2}^Q \mu_q (-L_q +1) + \sum_{q=2}^Q \nu_q (L_q - U_q).\label{eqLagrangian}
\end{align}
where ${L}_q$, $\mu_q$, $\nu_q \ge 0$, $\forall q\in[Q]$ and $\lambda \in \mathbb{R}$.

\begin{lemma}\label{theoremLagrange1}
	The optimal cache-allocation vector $\mathbf{L}$ for fixed $Q,\mathbf{n}$ is given by
	\begin{align}\label{eqGeneralL}
		L_{q} =
		\begin{dcases}
			1, & q\in\phi\!\cup\! \{1\}\\
			U_{q}, & q\in \psi\\
			\sqrt{ 	\frac{\pi_{q}}{n_{q}\!-\!n_{q-1}}}  \frac{{LN - n_{1} - \Phi_S - \Psi_S }}{\sum_{r\in\chi \sqrt{\pi_{r}( n_{q}-n_{q-1})}}}, & q\in \chi
		\end{dcases}
	\end{align}
where $\Phi_S = \sum_{q\in\phi}(n_{q}-n_{q-1})$, $\Psi_S = \sum_{q\in\psi}U_{q} \cdot(n_{q}-n_{q-1})$, and $\phi\cup \chi\cup \psi \cup \{1\} = [Q]$.
\end{lemma}

\begin{proof}
	The proof is relegated to Appendix~\ref{proofLagrange}.
\end{proof}

\begin{theorem}\label{theoDelay}	
	The expected delay optimized over $\mathbf{L}$ takes the form
	\begin{align}\nonumber
		T^{\star}(Q,\mathbf{n}) =& n_{1}  + \frac{K(1-\gamma)}{1+\Lambda\gamma}\sum_{q\in\phi}\pi_{q} + |\psi|\frac{\Lambda(1-\gamma)}{1+\Lambda\gamma} \\
		&+\frac{K(1-\gamma)}{1+\Lambda\gamma}\frac{ \left( \sum_{q\in\chi}\sqrt{\pi_{q}(n_{q}-n_{q-1}} \right)^{2}}{ LN-n_{1} - \Phi_{S} - \Psi_{S}}.
		\label{eqDelayOptL}
	\end{align}

\end{theorem}

\begin{proof}
	The proof is direct by inserting the calculated values $L_{q}$ from \eqref{eqGeneralL} into the expression of the expected delay \eqref{eqDeliveryTime2}.
\end{proof}

\begin{lemma}\label{lemmaRelaxationL}
	The continuous relaxation of $\mathbf{L}$ requires the use of memory sharing (cf.~\cite{maddah2014fundamental}).
	This would result in a performance loss that is bounded by a multiplicative factor of $1.12$.
\end{lemma}
\begin{proof}
	The optimal solution provided by Lemma~\ref{theoremLagrange1} may produce non-integer $L_q$.
	In order for the algorithm of \cite{lampirisSubpacketizationJSAC} to handle such non-integer values, we apply memory sharing as in \cite{maddah2014fundamental}.
	Specifically, each file with a non-integer cache-redundancy $L_q$ would be split into two parts, one part is cached with redundancy $\lceil L_q\rceil$, and the other part with $ \lfloor L_q \rfloor$.
	If we denote with $p\in[0,1]$ the fraction of the file stored with redundancy $ \lceil L_q\rceil $ we can calculate its value through
\begin{equation}\label{eqMemorySharing}
		p \lceil L_q\rceil + (1-p ) \lfloor L_q \rfloor = L_{q},
	\end{equation}
	The memory sharing approach invariably results in some loss in performance, but as we show promptly it remains small.
	
	Assuming that the target non-integer cache redundancy of sub-library $\mathcal{B}_{q}$ is $  L_q  + r$, $r<1$ i.e., $p = r$ by \eqref{eqMemorySharing}.	
	Focusing on the performance loss between the theoretical value (non-integer $L_q$) compared to the one achieved by memory sharing we have
	\begin{equation}\label{eqRatioMemShar}
		\frac  {  \frac{p}{  L_q +1}+  \frac{1-p}{  L_q }}{ \frac{1}{L_q} }  = 1 + \frac{ r(1-r) }{  L_{q}  (L_{q}+1)}.
	\end{equation}
We can see that the biggest gap in \eqref{eqRatioMemShar} occurs when $r = \frac{1}{2}$.
	It follows that the maximum difference between the delivery time achieved without memory sharing and after we apply the technique would be for $\lfloor L_q\rfloor=1 $ amounting to $< 12\%$, while for $\lfloor L_q\rfloor=2$ this would be $< 4\%$. Similar calculations show that for sub-libraries with even higher $L_{q}$ the performance loss due to memory sharing becomes negligible.
	
	Taking into consideration that only one sub-library can have cache-redundancy $\lfloor L_q\rfloor=1 $, it follows that the overall loss due to memory sharing is strictly less than $12\%$.
\end{proof}

\begin{remark}
	Equation \eqref{eqDelayOptL} can be simplified when $\phi = \psi = \emptyset$ to the following
	\begin{equation}
		T^{\star}(Q,\mathbf{n}) = n_{1}  +\frac{K(1\!-\!\gamma)}{1+\Lambda\gamma}\frac{ \left( \sum_{q=2}^{Q}\sqrt{\pi_{q}(n_{q}\!-\!n_{q-1}} \right)^{2}}{ LN-n_{1} }.
	\end{equation}
\end{remark}

\subsection{Optimizing $\mathbf{n}$\label{secOptimizeN}}

Using the objective function in~\eqref{eqDeliveryTime2}, i.e. optimized over variables $\mathbf{L}$, we can proceed to minimize it with respect to $\mathbf{n}$ for some instance of $Q$.
To this end, we propose a novel algorithm (Algorithm~\ref{algoFuncOptimizeSingleN}), which recursively optimizes each of the elements of $\mathbf{n}$.

\begin{remark}
Numerical evaluation of \eqref{eqDeliveryTime2} suggests that it is \emph{discrete convex}~\cite{Murota1998}.
This suggests that applying Algorithm \ref{algoFuncOptimizeSingleN} yields the optimal result.
We defer the formal proof of this statement to future work, due to considerable technical difficulty.
\end{remark}

\begin{algorithm}[h]
\caption{\textbf{update}$(n_{q})$}
\label{algoFuncOptimizeSingleN}
\KwIn{	$n_{1}, n_{2}, ..., n_{q-1}$, $Q$		}
Initialize: $S_{q} = \{ n_{q-1}+1, \  N\!+\!q\!-\!Q\} $ (Search space)\\

\While{	$ S_{q}(1) \neq  S_{q}(2)	$	}{
(Calculate delay using first search space point)
\begin{flalign*}
	& n_q = S_{q}(1)			&\\
	& \mathbf{n}^{\star}(q+1 : Q)= \textbf{update}(n_{q +1})&\\
	& T_{S_{q}(1)} = T\big( n_{1}, ..., n_{q-1}, S_{q}(1), \mathbf{n}^{\star}(q+1 : Q)\big) &
\end{flalign*}\\
(Calculate delay using second search space point)
\begin{flalign*}
	& n_q = S_{q}(2)			&\\
	& \mathbf{n}^{\star}(q+1 : Q)= \textbf{update}(n_{q +1})&\\
	& T_{S_{q}(2)} = T\big( n_{1}, ..., n_{q-1}, S_{q}(2), \mathbf{n}^{\star}(q+1 : Q)\big) &
\end{flalign*}\\
(Update search space with mid-point)
\begin{flalign*}
	&	s_a =\text{round} \left ( \frac{S_{q}(1) +  S_q(2)}{2}\right) &\\
	&s_b =\arg\min_{s \in  S_q }T_{s}&
\end{flalign*}
\begin{align}\label{eqUpdateSearchSpace}
	S_{q} =  \big\{ \min\{s_a, s_b\}, 	\max \{s_a, s_b\} \big\}
\end{align}\\
	}
\KwOut{ $\mathbf{n}(q:Q) = \left\{ S_{q}(1), n_{q+1}^{\star}, ..., n_{Q}^{\star} \right\} $}
\end{algorithm}

\paragraph{Intuition behind the algorithm}

The main idea behind our algorithm is based on the observation that one can easily compare the delivery time achieved by two vectors $\mathbf{n}({1})$ and $\mathbf{n}({2})$, where i) the first $q-1$ elements of these vectors are the same, ii) they differ in the $q$-th element, and iii) the remaining $Q-q$ elements are chosen such that to minimize the delivery time, given the $Q-q$ first elements.
In other words, we are interested in comparing the following vectors
\begin{align*}
	\mathbf{n}(1) &= \{ n_{1}, ..., n_{q-1}, n_q(1), n_{q+1}^{\star}(1), ..., n_{Q}^{\star}(1)\} \\		\mathbf{n}(2) &= \{n_{1}, ..., n_{q-1}, n_q(2), n_{q+1}^{\star}(2), ..., n_{Q}^{\star}(2)\}
\end{align*}
where $n^{\star}_{r}$ denotes the $r$-th element that, conditioned on all previous elements, produces the lowest delivery time.

Hence, by fixing the first $q-1$ elements and, at the same time, for each value of $n_{q}$ having access to the values of elements $\{q+1, ..., Q\}$ that produce the lowest delivery time, we can apply the bisection algorithm to optimize the value of $n_q$, by searching the discrete space between $n_{q-1}+1$ and $N-Q+q$.

\paragraph{Explanation of algorithm}
Our algorithm consists of a single recursive function that begins from the maximum search space for $n_{1}$, i.e. points $0$ and $N-Q$.

Initially, the algorithm creates a \textit{While} loop which stops when the search space is reduced to a single element. Inside the \textit{While} loop, the algorithm sets $n_1$ equal to the lower boundary of the search space and proceeds to calculate the optimal remaining $Q-1$ elements. To achieve this, it recursively calls $\textbf{update}(n_{2})$.

In the same spirit, \textbf{update($n_{2}$)} starts searching for the optimal $n_{2}$, conditioned on the value of $n_{1}$ that is given as input. To this end, the algorithm sets $n_{2} = n_{1}+1$ and recursively calls $\textbf{update}(n_{3})$. The recursive call of function \textbf{update} continues in the same manner until the last element, $n_{Q}$, is reached. At this point, since all previous elements are set (elements $1, .., Q-1$) the algorithm can perform a bisection in the discrete space and produce the optimal $n_Q$.

The bisection procedure for $n_Q$, given fixed $n_1, .., n_{Q-1}$, is done by calculating the delivery time achieved using the lower boundary point (Step~$3$), and then by calculating the delivery time achieved by the highest boundary point (Step~$4$). Then, the boundaries of the new search space would include the boundary of the previous search space that produced the smallest delay as well as the middle point of the old boundary.

When the optimal $n_Q$ is produced, the algorithm returns that value to $\mathbf{update}(n_{Q-1})$, which continues with the calculation of the delay for point $n_{Q-1}$.
Further, the algorithm seeks to calculate the delivery time when $n_{Q-1}$ is equal to the other boundary of its search space. Similarly to before, the algorithm needs to first optimize $n_Q$, and as a result calls $\textbf{update}(n_{Q})$. After this operation has produced the optimal $n_Q$ the algorithm calculates the delivery time corresponding to the higher boundary point of search space $S_{Q-1}$ and now is able to update the boundaries of the search space. The new boundaries of the search space are the middle point of the old search space and the boundary of the old search space which has produced the lowest delivery time. Due to the convexity of each point $n_q$, given that all previous points are the same, and that all following points are optimized, we can conclude that the new search space is reducing the delivery time.

\begin{theorem}
	The worst-case complexity of Algorithm~\ref{algoFuncOptimizeSingleN} is polynomial in $N$ and specifically is upper bounded by $\left(\log_{2} N\right)^{Q}$.
\end{theorem}
\begin{proof}
	By focusing on the amount of steps required to optimize element $n_1$ we can conclude that a maximum of $\log_2N$ calculations need to take place. Further, for each such iteration we need to calculate a maximum of $\log_2 N$ values of $n_2$. Continuing in the same manner for the remaining $n_{q}$, we can conclude that the maximum amount of iterations is bounded by $\left(\log_2 N\right)^{Q}$.
\end{proof}

\begin{remark}	
	It is intersting to note at this point that for the simulated environments (see Section~\ref{sec:numev}) the observed optimal value of the number of sub-libraries $Q^{\star}$ is relatively small, taking the maximum value of $Q^{\star}=3$.
	In other words, the overall complexity of designing the caches of the transmitters remains computationally feasible.
\end{remark}

\subsection{Optimizing the number of sub-libraries $Q$}

Equipped with Algorithm~\ref{algoFuncOptimizeSingleN}, which outputs the optimal library boundaries for an arbitrary $Q$, we need to search for $Q^{\star}$ such that
\begin{equation}
	Q^{\star} = \arg\min_{Q\in[N]} \mathbb{E} \{ T_{\mathbf{n}, \mathbf{L}}( Q ) \}.
\end{equation}

As we show in the following lemma, function $T( Q, \mathbf{n}^{\star} )$ is monotonous decreasing in the absence of~\eqref{eqLambdaConstraint}.

\begin{lemma}\label{lemmaDecreasingOverQ}
	The objective function of Problem~\ref{prb:mainprb}, in the absence of \eqref{eqLambdaConstraint}, is monotonous decreasing with respect to $Q$.	
\end{lemma}
\begin{proof}
	Let us assume some arbitrary $Q$, for which the optimal delivery time, optimized over $\mathbf{n}^{\star},\mathbf{L}^{\star}$ takes the form
	\begin{equation}
		T_{Q}(\mathbf{n}^{\star}_{Q},\mathbf{L}^{\star}_{Q}) =n_{1} +  \sum_{q=2}^{Q} \frac{K_{q} (1-\gamma)}{ L_{q}(1+\Lambda \gamma) }.
	\end{equation}
	We can transition to $Q+1$ sub-libraries and split the last sub-library into two sub-libraries, i.e.
	\begin{equation*}
		\mathbf{n}_{Q+1} = \{ n_{1}^{\star}(Q) , ..., n_{Q-1}^{\star}(Q), n_{Q}(Q+1), n_{Q+1}(Q+1)\}
	\end{equation*}
	and $\mathbf{L}_{Q+1} = \{ {L}_{1}^{\star}(Q), {L}_{1}^{\star}(Q), ..., {L}_{Q}^{\star}(Q), {L}_{Q}^{\star}(Q)\}$.	
	The above choice of variables $Q+1$, $\mathbf{n}_{Q+1}$, $\mathbf{L}_{Q+1}$ produces the same delivery time as $Q, \mathbf{n}^{\star}_{Q}, \mathbf{L}^{\star}_{Q}$, i.e.
	\begin{align}
		T_{Q+1}(\mathbf{n}_{Q+1}, \mathbf{L}_{Q+1}) = T_{Q}( \mathbf{n}^{\star}_{Q}, \mathbf{L}^{\star}_{Q})
	\end{align}
	Since increasing the number of sub-libraries leads to at least the same delivery time, it follows that the objective function is monotonous decreasing with respect to $Q$, when optimized over variables $\mathbf{n}$ and $\mathbf{L}$.
\end{proof}

Lemma~\ref{lemmaDecreasingOverQ} shows the monotonicity of the objective function in the absence of constraint \eqref{eqLambdaConstraint}. Conversely, by re-introducing the constraint we can guarantee that the objective function is monotonous increasing after point $Q^{\star}$.

Using the result of Lemma~\ref{lemmaDecreasingOverQ} we can see that a simple bisection algorithm in the discrete search space allows to successfully retrieve the optimal value of $Q$.

{
\subsection{Problem~\ref{prb:mainprb}'s relation to biconvex minimization problems}
Before moving on to the analysis of the performance of our proposed method, we would like to discuss the relationship between our Problem~\ref{prb:mainprb} and the biconvex minimization problems.
We begin by proving the biconvexity of our problem, as captured by the following lemma.
\begin{lemma}
Function~\eqref{mainoptfunc} is biconvex in $\mathbf{L}, \mathbf{n}$ for fixed $Q$.
\end{lemma}
\begin{proof}
The proof is detailed in Appendix~\ref{sec:BiconvexProof}.
\end{proof}

There are various methods and algorithms in the literature for solving biconvex minimization problems through exploitation of the biconvex structure of the problem~\cite{GorPfeKla2007}.
For instance, Alternate Convex Search (ACS) is a minimization method, derived as a special case of the Block-Relaxation Methods, where the variable set is divided into disjoint blocks~\cite{WenHur1976, Baz1993, deL1994}.
In each step, only one set of variables is optimized while the others remain fixed.
ACS does not provide any global optimality guarantee and the final solution may reach a local optimum or a saddle point.
The Global Optimization Algorithm (GOA), proposed in~\cite{FloVis1990, Flo2000}, aims to take advantage of the biconvex structure of the problem using a primal-relaxed dual approach, which can provide an upper bound and a lower bound to the optimal solution, thus further leading to a global optimality guarantee.
Obtaining an upper bound is done by solving the primal problem and is performed identically to the ACS approach, where a step optimizes the variables of a single variable set.
On the other hand, the lower bound is obtained by applying duality theory and linear relaxation.
The resulting relaxed dual problem is solved by considering every possible combination of bounds.
Iterating between the primal and the relaxed dual problem yields a finite $\epsilon$-convergence to the global optimum.

Even though the objective function given in~\eqref{mainoptfunc} is a biconvex function, it is easy to verify that constraint~\eqref{eqLambdaConstraint} is not convex when $\mathbf{n}$ are optimized for fixed $Q$ and $\mathbf{L}$. Therefore, Problem~\ref{prb:mainprb} does not satisfy $\text{Conditions}(\text{A})$ provided in~\cite{FloVis1990}, which points to the reason why our problem cannot be solved by applying GOA.
Further, using GOA in order to calculate a bound of our problem, would require the discarding of constraint~\eqref{eqLambdaConstraint}.
As we show in the next section (Section~\ref{secSpeedingUp}) discarding constraint~\eqref{eqLambdaConstraint} allows us to reach an analytical solution for the performance of our setting.

We need to note here that a setting where constraint~\eqref{eqLambdaConstraint} is always satisfied can be interpreted as one with a very high number of users, or more accurately a very high ratio $\frac{K}{\Lambda}$, and very high number of transmitters $K_{T}$.
In such a setting, as it will also become evident from the simulations (Section~\ref{sec:numev}), the achieved delay and the upper bound performance are becoming narrowly smaller.
}

\section{Performance analysis}\label{secSpeedingUp}

In this section we provide a bound on the expected achieved delivery time, and further prove that \textit{any} sub-library segmentation, as described by our main problem (Problem~\ref{prb:mainprb}), would yield a decreased expected delivery time compared to the uniform popularity case.

The bound is achieved by utilizing the outcome of Lemma~\ref{lemmaDecreasingOverQ}, describing the monotonicity of the objective function over variable $Q$, as well as expression~\eqref{eqNewObjectiveFunction}, obtained in the following lemma, describing the form of the objective function optimized with respect to $\mathbf{L}$.

\begin{lemma}\label{lemmaLagrange2}
	The optimal allocation of the cache-redundancy vector $\mathbf{L}$ for each sub-library, under the assumption that constraint \eqref{eqLambdaConstraint} is satisfied away from the boundaries, results in the objective function
	\begin{equation}\label{eqNewObjectiveFunction}
		T(Q, \mathbf{n}) = n_{1} + \frac{K(1-\gamma)}{1+\Lambda\gamma} \frac{\left( \sum_{q=2}^{Q} \sqrt{ \pi_{q}( n_{q}\!-\!n_{q-1}) }\right)^2}{LN-n_1}.
	\end{equation}
\end{lemma}
\begin{proof}
	Inserting the optimal cache-allocation calculated in \eqref{eqGeneralL} for $\phi = \psi = \emptyset$ into~\eqref{eqDeliveryTime2} yields the result.
\end{proof}

The main idea behind the performance bound is to utilize the monotonicity of the objective function with respect to $Q$, in the absence of constraint~\eqref{eqLambdaConstraint}, which leads to the conclusion that the expected delay is minimized when $Q=N$.

\begin{lemma}
	The minimum expected delivery time $\mathbb{E}\{ T^{\star} \}$ under the assumption of file popularity following a Zipf distribution with parameter $\alpha$ is lower bounded by
\begin{equation}
	\mathbb{E}\{ T^{\star} \}\ge \frac{K(1-\gamma)}{L N (1+\Lambda\gamma)} \frac{ \left( \sum_{q=1}^{N} {q^{-\alpha/2} } \right)^2 }{ \sum_{q=1}^{N}  {q^{-\alpha} } }
\end{equation}

Consequently, the maximum multiplicative ratio $G_{\max}$ that can be achieved by the optimal expected delay $T^{\star}$ compared to the delay of the uniform popularity case is bounded as
\begin{equation}\label{eqBoostTheoretical}
	G_{\max} \le  N  \frac{ \sum_{q=1}^{N} q^{-\alpha}  }{ \left( \sum_{q=1}^{N} q^{-\alpha/2}  \right)^2 }.
\end{equation}
\end{lemma}

\begin{proof}

In order to bound the optimal expected delivery time we remove constraint~\eqref{eqLambdaConstraint} and constraint $L_{q}\ge 1$.

Then, it follows from Lemma~\ref{lemmaDecreasingOverQ} that the minimum value of~\eqref{mainoptfunc} is achieved for $Q^{\star}=N$, which
implies $\mathbf{n} = [N]$, i.e. each sub-library is comprised of a single file. By incorporating the result of \eqref{eqNewObjectiveFunction} we can write the expectation of the objective function for $Q=N$ and $\mathbf{n} = [N]$ as
\begin{equation}\label{eqMaxPerformance}
	\mathbb{E}\{ T( N, [N]) \} = \frac{K(1-\gamma)}{L N (1+\Lambda\gamma)} \left( \sum_{q=1}^{N} \sqrt{p_{q}}\right)^2.
\end{equation}

Using that the fact that the file popularity follows the Zipf distribution with parameter $\alpha$, we can rewrite \eqref{eqMaxPerformance} as
\begin{equation}
	\mathbb{E}\{ T( N, [N])  \}= \frac{K(1-\gamma)}{L N (1+\Lambda\gamma)} \frac{ \left( \sum_{q=1}^{N} \frac{1}{q^{\alpha/2} } \right)^2 }{ \sum_{q=1}^{N} \frac{1}{q^{\alpha} } }
\end{equation}
The ratio between the above result and the uniform-popularity case, where the delivery time is $T_{\text{u}} = \frac{K(1-\gamma)}{L (1+\Lambda\gamma)}$, yields the result of \eqref{eqBoostTheoretical}.
\end{proof}

\begin{figure}
\centering
\includegraphics[width=0.8\columnwidth]{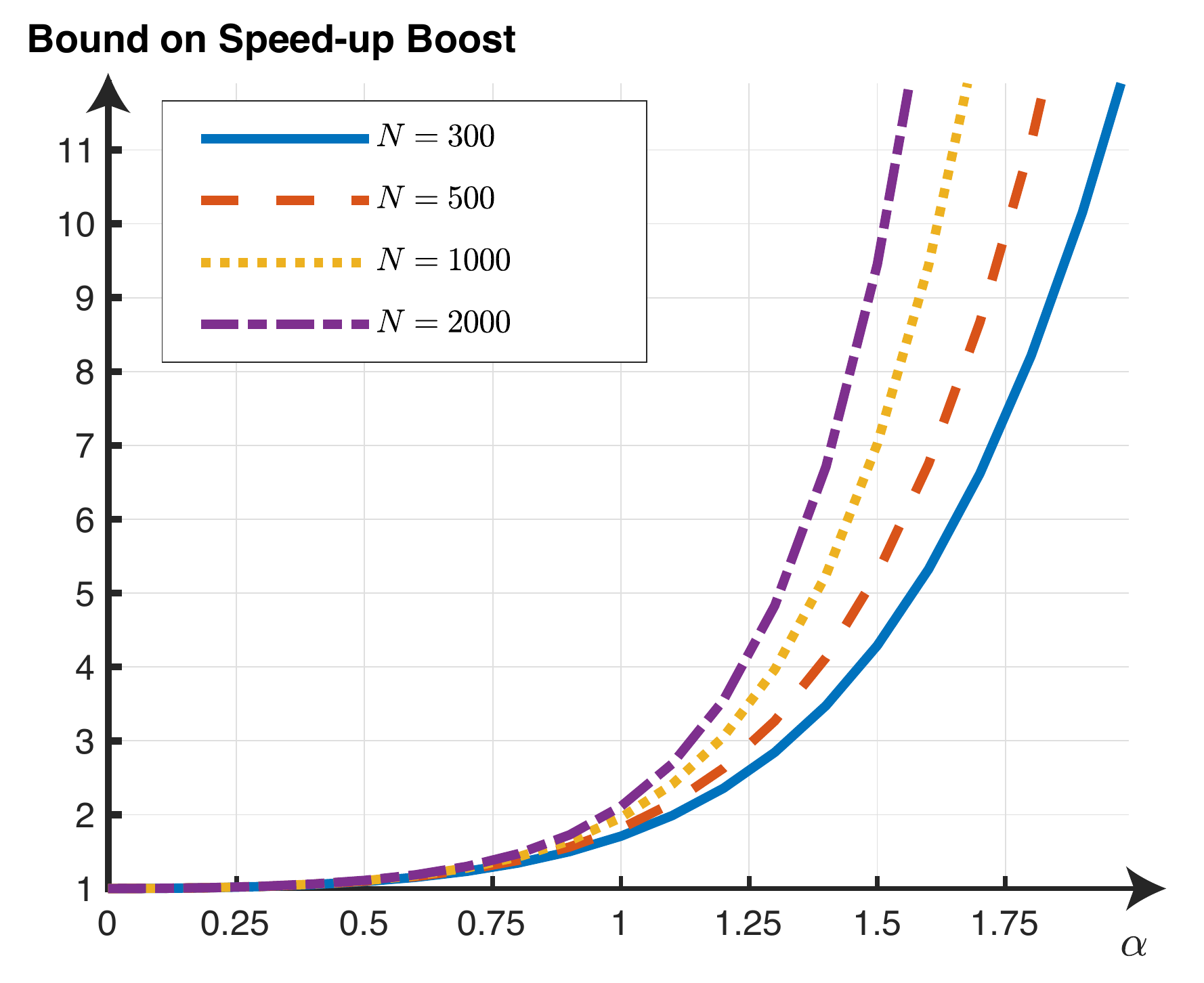}
\caption{The bound on the multiplicative boost,~\eqref{eqBoostTheoretical}, as a function of $\alpha$.}
\label{figTheoreticalOptimal}
\end{figure}

\begin{figure}
\centering
\includegraphics[width=0.72\columnwidth]{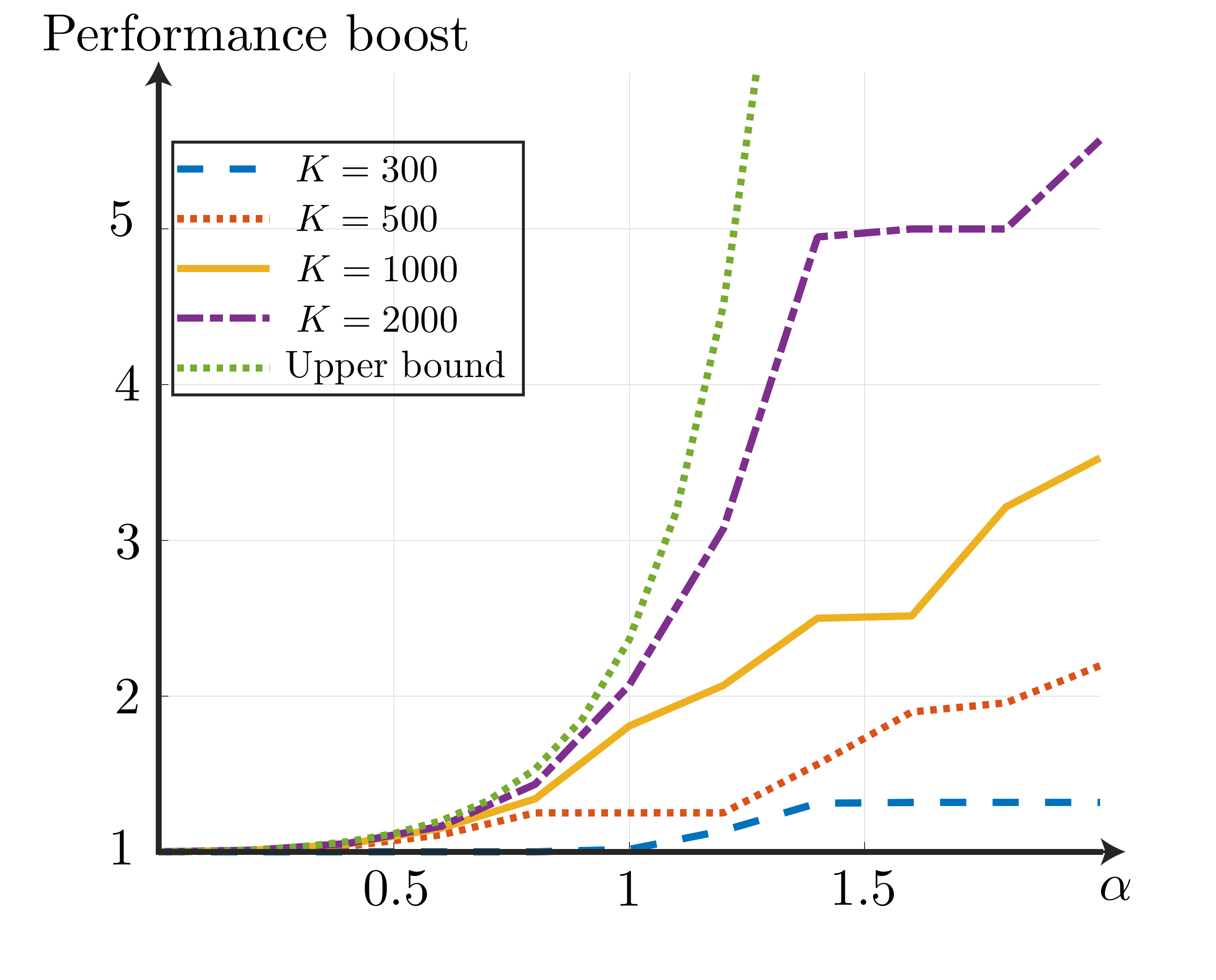}
\caption{The multiplicative boost of the expected performance achieved by our algorithm compared to the setting with uniform file popularity. The comparison is displayed here as a function of the Zipf parameter and for various $K$. The number of files across the examples is $N=6000$.}
\label{fig:res1}
\end{figure}
As we can see, the gain achieved is not depend on the number of users $K$. This is due to the lack of constraint \eqref{eqLambdaConstraint} which would otherwise enforce each cache-allocation variable $L_{q}\le \min\{  \frac{K_{q}}{\Lambda}, K_{T}\} $.
In Figure~\ref{fig:res1} we compare the theoretical result from \eqref{eqBoostTheoretical} with the numerical results of Section~\ref{sec:numev}.
It is interesting to note that as the number of users increases, the gains achieved in the simulations are moving closer to the theoretical bound. This can be attributed to the fact that as the number of users increases the cache-allocation variables $L_{q}$ are allowed to increase, in conjunction with constraint \eqref{eqLambdaConstraint}, hence the $L_{q}$ variables move closer to their optimal values.

We continue with a result that shows that any library segmentation, as long as constraint~\eqref{eqLambdaConstraint} is satisfied, would lead to a lower or equal delay compared to the uniform popularity case.

\begin{corollary}\label{corImprovement}
	Any library segmentation $\mathbf{n}$ that respects constraint~\eqref{eqLambdaConstraint} and is optimized over vector $\mathbf{L}$ improves upon the delivery time of the uniform popularity case i.e.
	\begin{align}
		\mathbb{E} \big\{ T(Q,  \mathbf{n}, \mathbf{L})  \big\} &< \frac{K (1-\gamma)}{L (1+ \Lambda\gamma)},\ \ 
\\
		L_{q}\! \le\! U_{q},\  \forall q\in[Q] \ \ \  \
		 & \mathbf{n}\in [N]^{Q} : n_{i}< n_{j}, \ i <j .  \nonumber
	\end{align}\end{corollary}
\begin{proof}
We consider some arbitrary library segmentation $\mathbf{n}$ which respects constraint~\eqref{eqLambdaConstraint}, and the objective function in~\eqref{eqNewObjectiveFunction}, i.e. after optimized over vector $\mathbf{L}$.
	\begin{align}\label{eqCS1}
	\mathbb{E} \{ T(Q, \mathbf{n}) \} &= \frac{K(1-\gamma)}{L (1+\Lambda\gamma)} \frac{\left( \sum_{q=1}^{Q} \sqrt{ \pi_{q}( n_{q}\!-\!n_{q-1}) }\right)^2}{N}\\ \label{eqCS2}
			&\!\le\! \frac{K(1\!-\!\gamma)}{L (1\!+\!\Lambda\gamma)} \frac{   \sum_{q=1}^{Q} (\sqrt{\pi_{q}})^2    \sum_{q=1}^{Q}( \sqrt{ n_{q}\!-\!n_{q\!-\!1}})^2}{N}\\
			&= \frac{K(1-\gamma)}{L  (1\!+\!\Lambda\gamma)}.
\end{align}
The transition from \eqref{eqCS1} to \eqref{eqCS2} makes use of the Cauchy-Schwartz inequality, where the first summation in \eqref{eqCS2} is equal to $1$, while the second summation is equal to $N$.
Thus, \textit{any} library segmentation $\mathbf{n}$, under the optimal cache-redundancy allocation, is upper bounded by the delivery time of the uniform popularity setting.

	Further, we can deduce the choices of $\mathbf{n}$ that do not improve the delivery time, compared to the uniform case. Specifically, we can view $\mathbb{E}\{ T(Q,\mathbf{n}) \}$ as the dot product of vectors $ \boldsymbol\pi_{1/2}\triangleq (\sqrt{\pi_1}, ..., \sqrt{\pi_{Q}})$ and $ \mathbf{n}_{1/2}\triangleq (\sqrt{n_1}, \sqrt{n_2-n_1}..., \sqrt{n_{Q}-n_{Q-1}})$.
	
	In order for the equality to hold in the Cauchy-Schwartz inequality, since neither $\boldsymbol{\pi}_{1/2}$ nor $\mathbf{n}_{1/2}$ can be the all zero vector, it is required that the two vectors are linearly dependent, i.e. $ {\boldsymbol\pi}_{1/2} = \lambda {\mathbf n}_{1/2}$, $\lambda\in\mathbb{R}$, \cite{apostol1974mathematical}. In other words,
	\begin{align} \label{eqCS3}
		\lambda^2 (n_{q} - n_{q-1}) &= \pi_q, \ \ \ \forall q \in [Q].
	\end{align}
Summing \eqref{eqCS3} over all $q$ yields $\lambda^2 \cdot N= 1$. Thus, for $\pi_{1}$ it must hold that
\begin{align*}
	\pi_{1} =  \frac{n_{1}}{N}
\end{align*}
which cannot be satisfied regardless of the sub-library segmentation when $\alpha>0$ and $Q>1$. Hence, any choice of $\mathbf{n}$ which satisfies constraint \eqref{eqLambdaConstraint}, would lead to an improved delivery time compared to the uniform-popularity case.
\end{proof}

\begin{remark}\label{remarkNonOptimalLibraries}
	Based on the result of Corollary~\ref{corImprovement}, we can see that this improvement would not necessarily hold in the general library segmentation considered in Problem~\ref{eqOriginalProblem}. Specifically, we can easily see that there are many library segmentations that satisfy ${ \boldsymbol\pi}_{1/2} =\frac{1}{N}\cdot {\mathbf{n}}_{1/2}$.
\end{remark}

\section{Numerical evaluation}\label{sec:numev}

\begin{table*}[t!]
\begin{centering}
\begin{tabular}{|c|c|c|c|c|c|c|c|c|c|}
   \hline
   \quad & \multicolumn{2}{|c|}{$K = 300$} & \multicolumn{2}{|c|}{$K = 500$} & \multicolumn{2}{|c|}{$K = 1000$} & \multicolumn{2}{|c|}{$K = 2000$}\\
  \hline
   $\alpha$ & $\mathbf{n^*}$ & $\mathbf{\bar{L}}$ & $\mathbf{n^*}$ & $\mathbf{\bar{L}}$ & $\mathbf{n^*}$ & $\mathbf{\bar{L}}$ & $\mathbf{n^*}$ & $\mathbf{\bar{L}}$\\
  \hline
$0.2$ & $[0]$ & $[5.0000]$ & $[0, 2174]$  & $[5.5405, 4.6929]$  & $[0, 1004, 2591]$  & $[5.9481, 5.1025, 4.6726]$  & $[0, 466, 2138]$  & $[6.4407, 5.2914, 4.6998]$\\
  \hline
$0.4$ & $[0]$ & $[5.0000]$ & $[0, 1923]$  & $[6.2933, 4.3900]$  & $[0, 817, 2530]$  & $[7.4944,5.2183, 4.3049]$  & $[0, 353, 1907]$  & $[8.9654, 5.7116, 4.3878]$\\
  \hline
$0.6$ & $[0]$ & $[5.0000]$ & $[0, 1678]$  & $[7.3849, 4.0740]$  & $[0, 634, 2262]$  & $[9.8119, 5.4960, 3.9679]$  & $[0, 251, 1652]$  & $[13.1689, 6.3736, 4.0858]$\\
  \hline
   $0.8$ & $[0]$ & $[5.0000]$ & $[0, 1431]$  & $[8.8101, 3.6899]$  & $[0, 785]$  & $[22.6643, 2.3357]$  & $[0, 191, 1439]$  & $[20.4944, 7.1939, 3.7507]$\\
  \hline
   $1$ & $[1]$ & $[5.0007]$ & $[0, 1582]$  & $[10.7041, 1.7959]$  & $[0, 550]$  & $[13.7520, 4.1168]$  & $[0, 157, 1278]$  & $[30.3803, 8.1446, 3.4096]$\\
  \hline
$1.2$ & $[1]$ & $[5.0007]$ & $[0, 816]$  & $[11.3591, 1.1409]$  & $[0, 490]$  & $[18.2515, 3.8216]$  & $[0, 233]$  & $[41.1398, 3.6763]$\\
  \hline
$1.4$ & $[1]$ & $[5.0007]$ & $[3]$  & $[5.0020]$  & $[0, 400]$  & $[23.7579, 1.2417]$  & $[0, 212]$  & $[46.6118, 3.3188]$\\
  \hline
$1.6$ & $[1]$ & $[4.2058]$ & $[2]$  & $[5.0013]$  & $[5]$  & $[5.0033]$  & $[0, 98]$  & $[47.8672, 2.1325]$\\
  \hline
$1.8$ & $[1]$ & $[3.5129]$ & $[2]$  & $[3.9464]$  & $[4]$  & $[4.9572]$  & $[0, 15]$  & $[46.3237, 3.6761]$\\
  \hline
$2$ & $[1]$ & $[2.9401]$ & $[1]$  & $[4.9001]$  & $[3]$  & $[4.3115]$  & $[5]$  & $[5.0033]$\\
  \hline
  \noalign{\vskip 1mm}
\end{tabular}
\caption{Optimal sub-library boundaries $\mathbf{n}^{\star}$ and optimal antenna allocations $\mathbf{L}^{\star}$. The number of sub-libraries is given by $Q = 1+|\mathbf{n}^{\star}|$. When the first value of $\mathbf{n}^{\star}$ is $0$ it points to an empty $\mathcal{B}_{1}$ sub-library.}
\label{tab:tab1}
\end{centering}
\end{table*}

\begin{figure}
\centering
\includegraphics[width=0.95\columnwidth]{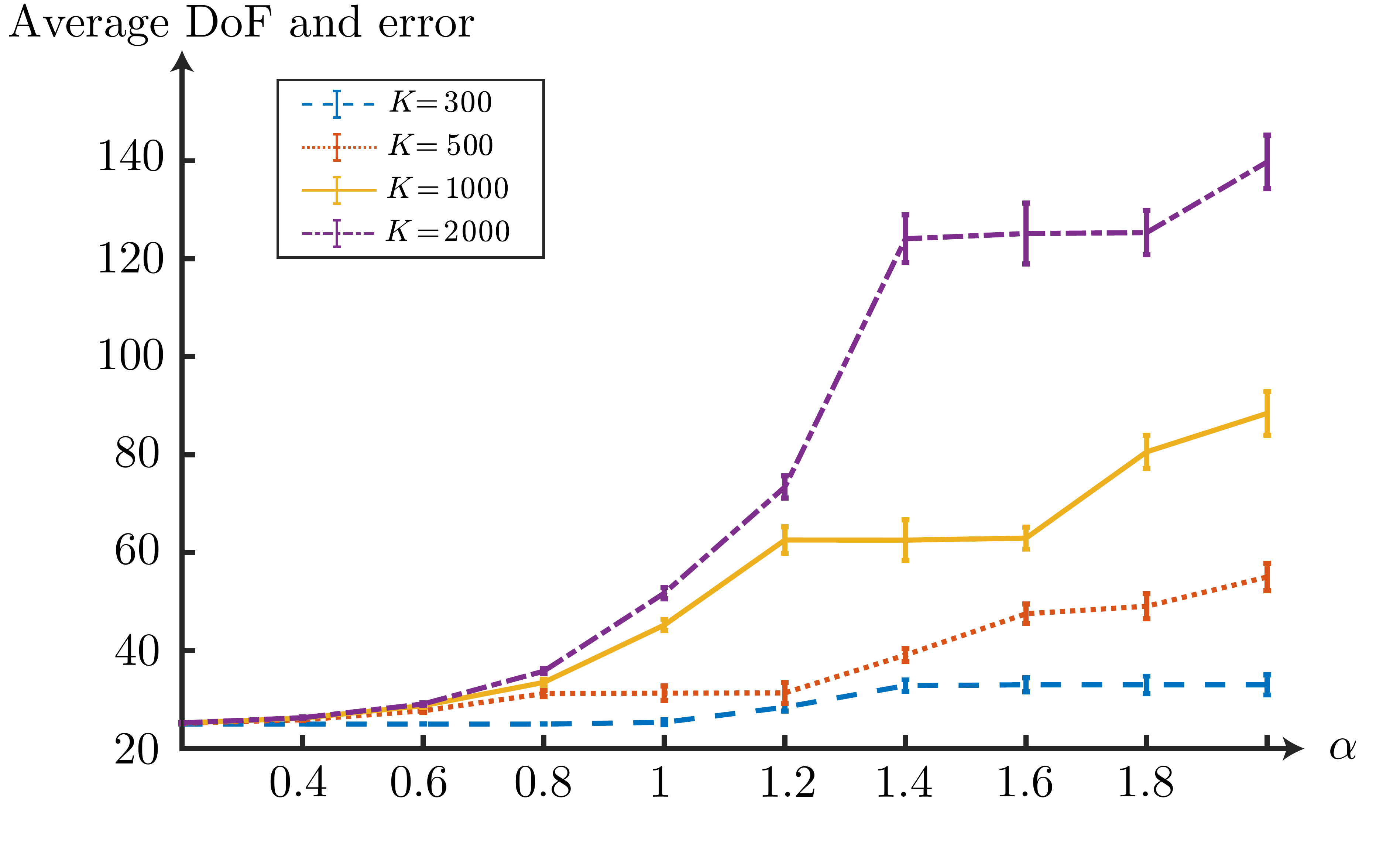}
\caption{
Expected DoF and standard deviation achieved under the placement dictated by the algorithm of Section~\ref{secSolutionOverview}. User preferences are drawn according to the Zipf distribution.
The deviation from the mean is  remains less than $1$ DoF for small values of $\alpha$, while for bigger values of $\alpha$ and number of users the deviation is approximately $5-10\%$ of the achieved DoF.}
\label{figErrors}
\end{figure}

To illustrate the performance of our proposed placement, we consider two scenarios that differ on the library sizes and the caching capabilities of the transmitters and the users.
The first scenario focuses on a library with TV series, comprised of many files, but each of relatively small size, thus allowing a higher percentage of the library to be stored at the receivers.
On the second scenario, we have a library of movies which, although has much fewer individual files, nevertheless each file has higher size.

\paragraph*{Scenario $1$}

We consider a typical, dense multiple transmitter setting \cite{7248710,6786390,poularakisDataOffloadingTNSM2016,shanmugamFemtocaching2013TransIT}, where a set of $K_{T} = 50$ single-antenna transmitters are connected to $K$ receivers.
The content library is comprised of $N=6000$ TV series episodes, such as typically found in the Netflix catalogue of European countries \cite{Batikas2015Film}.
The size of each such episode is assumed to be $100$MB, i.e. of standard definition quality, while its duration approximately $45$min.
Each transmitter and each receiver can store $10\%$ of the whole library, i.e. $\gamma_{T} = \gamma = \frac{1}{10}$ which amounts to $60$GB.
For a packet size of $1$KB the subpacketization is constrained to be $F \le 10^5$, thus the maximum number of different caches allowed is $\Lambda = 40$ ($\binom{40}{4} \approx 9\cdot 10^4$).

In Figure~\ref{fig:res1} we plot, for varying number of users $K\in \{300, 500, 1000, 2000\}$, the ratio of the delivery time of the non-uniform setting over the expected delivery time of our scheme as well as the upper bound calculated in Section~\ref{secSpeedingUp}, as a function of parameter $\alpha$.
We observe that for $\alpha= 0.8$ the delay reduction, compared to the uniform popularity case, ranges between $25\%$ (factor $1.3$) and $45\%$ (factor $1.8$) for $K=500$ and $K=2000$, respectively and further increases to a multiplicative factor of $2.8$ for $\alpha = 1.2$.
Another important point is that for $\alpha \le 1.2$ the proposed scheme remains close to the upper bound.

Further, in Figure~\ref{figErrors} we plot the average DoF performance as a function of $\alpha$ for all the values of $K$ of our example, as well as the deviation from the mean produced by $10^3$ simulations.
We note that for practical values of parameter $\alpha$ ($\alpha \le 1.2$), the DoF performance varies slightly from the mean value.

The optimal sub-library boundaries $\mathbf{n}^{\star}$ and the optimal cache-allocation values $\mathbf{L}^{\star}$ for each of the parameters of Scenario $1$ are displayed in Table~\ref{tab:tab1}.

\paragraph*{Scenario $2$}
	Let us now consider another network that aims to serve content from a library of $N=3000$ movies, typical of a Netflix catalogue \cite{Batikas2015Film}, each of size of $1$GB, of average duration $1.5$h, and of standard definition quality.
	User demands are satisfied by a set of $K_T = 20$ single-antenna transmitters.
	Due to the much higher per-file size, the normalized cache of a user is $\gamma = \frac{1}{50}$, while we consider that each transmitter's cache is $\gamma_T = \frac{1}{10}$. Hence, a user dedicates $60$GB for caching, while a transmitter dedicates $300$GB.
	Assuming, as before, that the minimum packet size is $1$KB, translates to a maximum supacketization of $F\le 10^6$ packets thus, the maximum number of different caches allowed is $\Lambda = 150$ ($\binom{150}{3} \approx 5,5\cdot 10^{5}$).
	
The performance boost, compared to the non-popularity case, is displayed in Figure~\ref{fig:res2} for varying number of users $K = \{ 500, 1000, 2000\}$.
It is interesting to note that as the number of users increase one can get close to the bound.

	\begin{figure}
\centering
\includegraphics[width=0.75\columnwidth]{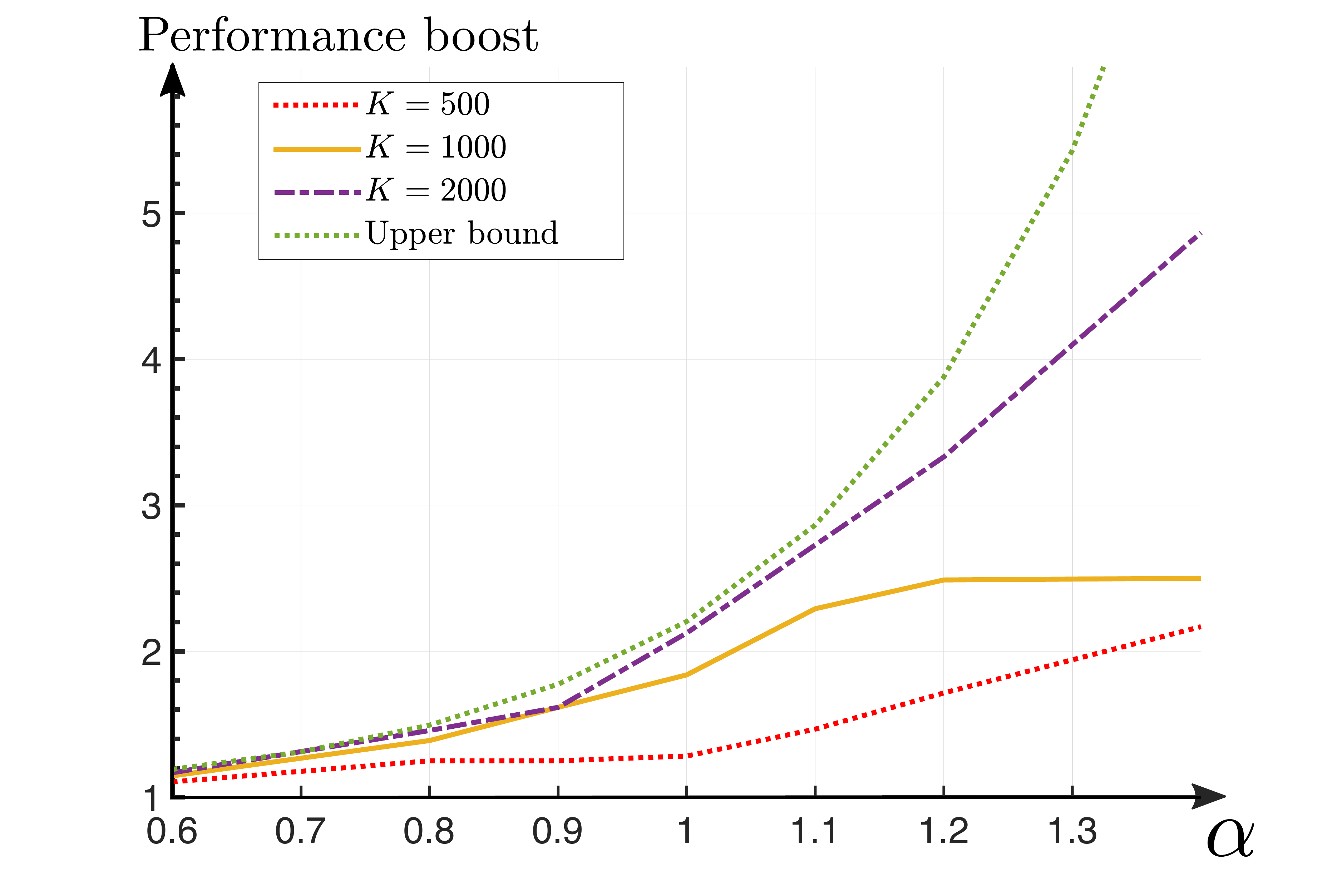}
\caption{Scenario $2$. The multiplicative boost of the expected performance achieved by our algorithm compared to the setting with uniform file popularity. The comparison is displayed here as a function of the Zipf parameter, for various numbers of users $K$. The library has $N=3000$ files.}
\label{fig:res2}
\end{figure}

\section{Final remarks and Conclusions}\label{sec:conc}
Our work showed for the first time how one can leverage file popularity in order to optimize the cached content at multiple transmitters and achieve multiplicative increase in the performance of coded caching systems. This performance increase can occur even when the file popularity is not very skewed, and it can occur in the subpacketization-constrained regime, where it is indeed needed the most. 
While most single antenna cache-aided systems exhibit success in increasing the ``usable'' part of a user's cache when exploiting file-popularity at the receiver side, we have showed here how multi-transmitter environments can provide multiplicative gains by becoming popularity-aware, while not affecting the structural symmetry in which coded caching thrives.

\bibliographystyle{ieeetr}

\appendices

\section{Transmitter-side caching policy\label{sec:Redundancy}}
To implement the cache-redundancy allocation of our algorithm, for any values of parameters $Q, \mathbf{n},\mathbf{L}$, we extend the approach of \cite{lampirisSubpacketizationJSAC} to account for multiple sub-libraries with different redundancies. We note that the objective is to place each file of sub-library $\mathcal{B}_{q}$ in exactly $L_{q}$ different transmitters. The proposed placement paired with the delivery process we described in the main part of the document can produce the DoF performance of \eqref{eq:qDelay}, as illustrated in \cite{lampirisSubpacketizationJSAC}.

The placement is done sequentially.  We start from the first sub-library and we consecutively cache the whole first file into the first $L_{1}$ transmitters, then the second file (of the first sub-library) into transmitters $L_{1}+1$ through $1+(2 L_{1}-1 \mod K_{T})$, and so on for the remaining files of $\mathcal{B}_{1}$. The selection of the transmitters is always done using the modulo operation, which means that when we place a file at the last transmitter, we continue the process with the first transmitter.

After storing each file from the first sub-library in a total of $L_{1}$ transmitters each, we proceed with the second sub-library. Continuing from the transmitter after the one last used, i.e. continuing from transmitter $1+(n_{1}\cdot L_{1}\mod K_{T})$, we again sequentially fill the caches, starting from the first file of the second sub-library, which we now store in $L_{2}$ consecutive transmitters, and so on. The process is repeated for each sub-library $\mathcal{B}_{q}$, using the corresponding $L_q$, starting every time from the transmitter after the one last used.

Overall, the above process stores each file of sub-library $\mathcal{B}_q$ in exactly $L_{q}$ distinct transmitters. Further, through this cyclic assignment of files into transmitters we can guarantee that the cache-size constraint is satisfied, i.e. that each transmitter stores exactly $\gamma_{T}N$ files.

\section{Proofs of Section~\ref{secSolutionOverview}}
\subsection{Proof of Corollary~\ref{corMoreNeedsMore}}\label{proofMoreNeedsMore}

The delay required to satisfy solely the demands corresponding to sub-libraries $q,r$ can be written, after normalization by $\frac{K(1-\gamma)}{1+\Lambda\gamma}$, as
	\begin{equation}\label{eqPartialDeliveryTime}
		T_{\text{p}}(L_{q},L_{r}) = \frac{\pi_{q}}{L_{q}} +  \frac{\pi_{r}}{L_{r}}.
	\end{equation}
Using an equal cache-allocation, $L_{q} = L_{r} = \tilde{L}$, yields
	\begin{equation}
		T_{\text{p}}(\tilde{L},\tilde{L}) = \frac{\pi_{q}}{\tilde{L}} +  \frac{\pi_{r}}{\tilde{L}}.
	\end{equation}
	In contrast, if we assume that the two cache-allocations differ by $\ell$ such that $\ell < \tilde{L}$, we have
	\begin{align}
		T_{\text{p}}(\tilde{L}\!+\!\ell,\tilde{L}\!-\!\ell) &= \frac{\pi_{q}}{\tilde{L}+\ell} +  \frac{\pi_{r}}{\tilde{L}-\ell} \\
		&= \frac{(\pi_{q}-\pi_{r})\tilde{L}}{\tilde{L}^2-\ell^2} - \frac{(\pi_{q}-K_{r})\ell}{\tilde{L}^{2}-\ell^{2}}\\
		& < \frac{(\pi_{q}-\pi_{r})}{\tilde{L}} - \frac{(\pi_{q}-\pi_{r})\ell}{\tilde{L}^{2}} < T_{\text{p}}(\tilde{L},\tilde{L})\nonumber
	\end{align}
	which shows that it is always a better strategy to allocate higher cache redundancy to sub-libraries with higher cumulative probability.\hfill\qedsymbol

\subsection{Proof of Lemma~\ref{lemmaConsecutive}}\label{proofConsecutive}

We assume that $\{ \mathcal{B}_{q}\}_{q=1}^{Q} $ is an optimal library segmentation where, without loss of generality, $\pi_{q}\ge \pi_{q+1}, \forall q\in[Q\!-\!1]$.
In a different case we can rename the sub-libraries such that $\pi_{q}\ge \pi_{q+1}, \forall q\in[Q-1]$.
We pick two files, $W^{r_a}$, $W^{r_{b}}$, with corresponding popularity $p_{r_a}> p_{r_b}$, such that $W^{r_a}\in \mathcal{B}_{a}$ and $W^{r_b}\in \mathcal{B}_{b}$, while $\pi_{a} < \pi_{b}$, else we wouldn't have anything to prove.
Further, assuming that ${L}_{a}$ and $L_{b}>L_{a}$ correspond to the optimal cache-allocation of $\mathcal{B}_{a}$ and $\mathcal{B}_{b}$, respectively we can calculate the expected delay, $T_{1}$, of this sub-library segmentation and cache-allocation using \eqref{eqDeliveryTime1}.

Now, we can proceed to calculate the delay, $T_{2}$ of a similar system with the same cache-allocation as before, but now files $W^{r_a}, W_{r_{b}}$ are swapped, i.e. $W^{r_a}\in \mathcal{B}_{b}$ and $W^{r_b}\in \mathcal{B}_{a}$.

The difference of the two delays then takes the form
\begin{align}
	T_{1} - T_{2} &= \frac{p_{r_a} }{L_{a}} +  \frac{p_{r_b} }{L_{b}}  - \left( \frac{p_{r_b} }{L_{a}} +  \frac{p_{r_a} }{L_{b}} \right)\\
	&= ( p_{r_a} - p_{r_b}) \frac{L_{b}-L_{a}}{L_b \cdot L_a} >0.
\end{align}

Using the above result, and beginning from some arbitrary segmenation of the library we can select pairs of files which belong in different sub-libraries such that the probability of one file is higher than the probability of the other, while the more popular file resides in the less popular sub-library and swap them.
As we showed, performing this task will always transition the system to a lower delivery time.

Continuing to perform this task would result in a library segmentation where each sub-library has files of consecutive indices.\hfill\qedsymbol

\subsection{Convexity of~\eqref{eqDeliveryTime2} for fixed $Q,\mathbf{n}$\label{sec:BiconvexProof}}

{

Let us define the Hessian of $T^{\star}_{Q}$ with respect to $\mathbf{L}$ by $\mathbf{H}_1$, which is given by
\begin{align}\label{eqHessian}
\mathbf{H}_1 &= \begin{bmatrix}
\frac{\partial^2 T^*(Q)}{\partial L_1^2} & \frac{\partial^2 T^*(Q)}{\partial L_1 L_2} & \dots & \frac{\partial^2 T^*(Q)}{\partial L_1 L_Q}\\
\frac{\partial^2 T^*(Q)}{\partial L_2 L_1} & \frac{\partial^2 T^*(Q)}{\partial L_2^2} & \dots & \frac{\partial^2 T^*(Q)}{\partial L_2 L_Q}\\
\vdots & \vdots & \ddots & \vdots\\\
\frac{\partial^2 T^*(Q)}{\partial L_Q L_1} & \frac{\partial^2 T^*(Q)}{\partial L_Q L_2} & \dots & \frac{\partial^2 T^*(Q)}{\partial L_Q^2}
\end{bmatrix}.
\end{align}

Focusing on the $q$-th diagonal element of $\mathbf{H}_1$ we have
\begin{align}
	\frac{\partial^2 T^{\star}_{Q}}{\partial L_q^2} \!=\! \frac{\partial^2}{\partial L^{2}_{q}} {\bigg(n_1\! +\! \sum_{r = 2}^Q \frac{K(1-\gamma)\pi_r}{L_r(1+\Lambda\gamma)}\bigg)}
= 2 \frac{K\pi_{q}(1-\gamma)}{L_{q}^3(1+\Lambda\gamma)}\!>\!0\label{eq:hes1last}.
\end{align}

Similarly, we can show that the non-diagonal elements of $\mathbf{H}_1$ are equal to $0$.
Let us consider  arbitrary element $(q,s)$, $q\ne s$ for which we have
\begin{align*}
\frac{\partial^2 T^{\star}_{Q}}{\partial L_q \partial L_{s}} &= \frac{\partial^2}{\partial L_{q} \partial L_{s}} {\bigg(n_1 + \sum_{r=2}^{Q} \frac{K \pi_r (1-\gamma)}{L_{r}(1+\Lambda\gamma)}\bigg)}
\\
&= \frac{\partial}{\partial L_{s}} \bigg(-\frac{K \pi_q(1-\gamma)}{L_{q}^2(1+\Lambda\gamma)}\bigg)= 0.
\end{align*}

We can now conclude that $\mathbf{H}$ is positive semi-definite since its diagonal elements are positive, while its non-diagonal elements are zero.
Thus, function $T^{\star}_{Q,\mathbf{n}}$ is convex in $\mathbf{L}$.

Now, we continue with proving the convexity of $T^*(Q)$ in $\mathbf{n}$ for fixed $\mathbf{L}$. Since $p_j$ is defined only for discrete $j$, we replace the Zipf distribution with a continuous Pareto distribution,
\begin{equation*}
f(j) = j^{-\alpha} \underbrace{H(N,\alpha)^{-1}}_{\triangleq C},
\end{equation*}
where $H(N,\alpha)$ is the generalized Harmonic number.

Let us define the Hessian of $T^*(Q)$ with respect to $\mathbf{n}$ as $\mathbf{H}_2$, which is given similarly to \eqref{eqHessian}.

Following similar arguments used in showing the positive semi-definiteness of $\mathbf{H}_1$, we will show that $\mathbf{H}_2$ is also positive semi-definite.
It is trivial to show that first diagonal element of $\mathbf{H}_2$ is always non-negative. Let us consider a different diagonal element $q>1$, for which we get
\begin{align*}
\frac{\partial^2 T^{\star}_{Q}}{\partial n_q^2} &= -\frac{KC(1-\gamma)~ \alpha ~n_2^{-(\alpha+1)}}{1+\Lambda\gamma}\left[\frac{1}{L_q} - \frac{1}{L_{q+1}}\right]
&\geq 0,
\end{align*}
due to the fact that $L_q \geq L_{q+1},~\forall q \in [2,Q]$.

Since function~\eqref{eqDeliveryTime2} is a linear combination of terms, where each term is solely dependent on one of the $L_{q}$ variables it follows that a double partial differentiation over different $L_q, L_{k}$ would produce $0$.
Therefore, we conclude that $\mathbf{H}_2$ is positi0ve semi-definite and function $T^{\star}_{Q}$ is convex in $\mathbf{n}$ for fixed $\mathbf{L}$. Therefore, the function in \eqref{eqNewObjectiveFunction} is biconvex.

Finally, all the constraints are affine. Thus, Problem~\ref{prb:mainprb} for fixed $Q$ is a biconvex problem. \hfill\qedsymbol
}

\subsection{Proof of Lemma~\ref{theoremLagrange1}}\label{proofLagrange}

The output of the KKT conditions provides three different, and non-overlapping, subsets of $[Q]$.
The first set, $\phi$, is comprised of those $q\in[Q]$ for which $L_q=1$ (apart from $q=1$ for which, by definition, $L_1=1$).
The second set, $\psi$, is comprised of those $q$ for which $L_{q} = U_{q}$.
Finally, the remaining $q$ are contained in set $\chi$ and for these we need to calculate the cache-allocation variable $L_q$.

Hence, the expected delay can be written as
\begin{align}\nonumber
	T^{\star}(Q,\mathbf{n}) =& n_{1}  + \frac{K(1-\gamma)}{1+\Lambda\gamma}\sum_{q\in\phi}\pi_{q} + |\psi|\frac{\Lambda(1-\gamma)}{1+\Lambda\gamma} \\
		&+\frac{K(1-\gamma)}{1+\Lambda\gamma}\sum_{q\in\chi}\frac{\pi_{q}}{L_{q}}
\label{eqDelayPartOpt}.
\end{align}

The cache capacity constraint becomes
\begin{align}\nonumber
	&\sum_{q\in \chi } L_{q}(n_{q}-n_{q-1}) = \\
	&LN - \bigg[ n_{1}\! +\! \sum_{q \in \phi} (n_{q}\!-\!n_{q-1}) + \sum_{q\in \psi} U_{q}( n_{q}\!-\!n_{q-1}) \bigg]. \label{eqCapacityConst1}
\end{align}

Taking the derivative of~\eqref{eqLagrangian} with respect to $L_q$ such that $q\in\chi$, and equating it to $0$ yields
\begin{align}\label{eqPartialLagrangian}
	\frac{\partial \mathcal{L}}{\partial L_{q}} = -\frac{K \pi_{q}(1-\gamma)}{(1+\Lambda\gamma) L_{q}^{2}} + \lambda (n_{q}-n_{q-1}) = 0.
\end{align}
Separating $L_q$ from the remaining terms in \eqref{eqPartialLagrangian} yields
\begin{align}
	L_{q} &= \frac{1}{\sqrt{\lambda}} \sqrt{ \frac{ K(1-\gamma) \pi_{q}}{ (1+\Lambda\gamma)(n_{q}-n_{q-1})}}\\
	L_{q}(n_{q}-n_{q-1}) &= \frac{1}{\sqrt{\lambda}} \sqrt{ \frac{ K(1-\gamma) \pi_{q}(n_{q}-n_{q-1})}{ (1+\Lambda\gamma)}}.\label{eqDerivLq1}
\end{align}

Further, from \eqref{eqPartialLagrangian} keeping on one side terms $L_q, (n_q-n_{q-1}), \lambda$ we get
\begin{align}\label{eqDerivLq2}
	L_{q}(n_{q}-n_{q-1}) \lambda = \frac{K(1-\gamma)\pi_{q}}{(1+\Lambda\gamma)L_{q}}.
\end{align}

Summing \eqref{eqDerivLq1} and \eqref{eqDerivLq2} over all $q\in\chi$
\begin{align}\label{eqSumDerivLq1}
	\stackrel{\eqref{eqDerivLq1}}{\Rightarrow} \sum_{q\in\chi} L_{q}(n_{q}-n_{q-1}) &=  \frac{1}{\sqrt{\lambda}} \sum_{q\in\chi} \sqrt{ \frac{ K(1-\gamma) \pi_{q}(n_{q}-n_{q-1})}{ (1+\Lambda\gamma)}}\\
	\label{eqSumDerivLq2}
	\stackrel{\eqref{eqDerivLq2}}{\Rightarrow} \sum_{q\in\chi} L_{q}(n_{q}-n_{q-1}) &=\frac{1}{\lambda} \frac{K(1-\gamma)}{(1+\Lambda\gamma)}\sum_{q\in\chi} \frac{\pi_{q}}{L_{q}}.
\end{align}

Replacing $\lambda$ from \eqref{eqSumDerivLq2} to \eqref{eqSumDerivLq1} yields
\begin{align}
	\sum_{q\in\chi} L_{q}(n_{q}-n_{q-1}) &= \frac{ \frac{K(1-\gamma)}{1+\Lambda\gamma} \sum_{q\in\chi} \sqrt{ \pi_{q}(n_{q}-n_{q-1})}}
	{\frac{K(1-\gamma)}{(1+\Lambda\gamma)}\sum_{q\in\chi} \frac{\pi_{q}}{L_{q}}}\\
	\frac{K(1\!-\!\gamma)}{(1+\Lambda\gamma)}\sum_{q\in\chi} \frac{\pi_{q}}{L_{q}} &=
	\frac{ \frac{K(1-\gamma)}{1+\Lambda\gamma} \sum_{q\in\chi} \sqrt{ \pi_{q}(n_{q}\!-\!n_{q-1})}}
	{\sum_{q\in\chi} L_{q}(n_{q}\!-\!n_{q-1})}.
	\label{eqLambdaDivisions1}
\end{align}

Replacing the left-hand-side of \eqref{eqLambdaDivisions1} with \eqref{eqDelayPartOpt} and the denominator of \eqref{eqLambdaDivisions1} from \eqref{eqCapacityConst1} yields the result of Theorem~\ref{theoDelay}.\qedsymbol

\end{document}